\documentclass[
  aps,
  prx,
  reprint,
  superscriptaddress,
  bibnotes,
  longbibliography
]{revtex4-2}

\usepackage[utf8]{inputenc}
\usepackage[american]{babel}
\usepackage[normalem]{ulem}
\usepackage{mathpazo}
\usepackage{lmodern}
\usepackage[protrusion=true,expansion=true]{microtype}
\usepackage{physics}
\usepackage{graphicx}
\graphicspath{{Graph/}{./Graph/}}
\usepackage{indentfirst}
\usepackage{braket}
\usepackage{amsmath}
\usepackage{amssymb}
\usepackage{amsthm}
\usepackage{capt-of}
\usepackage{booktabs}
\usepackage{placeins}
\usepackage[title]{appendix}
\usepackage[dvipsnames]{xcolor}
\usepackage[unicode=true,bookmarks=false,breaklinks=true,colorlinks=true,linkcolor=blue,citecolor=blue,urlcolor=blue]{hyperref}
\usepackage{tikz}
\usetikzlibrary{arrows.meta,calc,fit,decorations.pathreplacing,trees}
\usepackage{orcidlink}

\usepackage[most]{tcolorbox}
\newcounter{protocol}
\newenvironment{protocolcard}[1]{%
  \refstepcounter{protocol}%
  \begin{tcolorbox}[
    enhanced,
    breakable=false,
    colframe=OliveGreen,
    colback=OliveGreen!2,
    boxrule=0.55pt,
    arc=1.2pt,
    left=6pt,right=6pt,top=5pt,bottom=5pt,
    before upper={\small},
    before skip=7pt,after skip=7pt
  ]%
  \noindent{\normalsize\bfseries
  \textcolor{OliveGreen}{Protocol~\theprotocol. #1}}\par\vspace{2pt}%
}{\end{tcolorbox}}
\newcommand{\protocolstep}[1]{%
  \par\smallskip\textbf{\color{OliveGreen}#1}\par\vspace{1pt}}

\newcommand{\iu}{\mathrm{i}}

\newcommand{\Real}{\operatorname{Re}}
\newcommand{\Imag}{\operatorname{Im}}
\newcommand{\E}{\mathbb{E}}
\newcommand{\Prob}{\mathbb{P}}
\newcommand{\ctrl}{\operatorname{ctrl}}
\newcommand{\Dph}{D_{\mathrm{ph}}}
\newcommand{\Uchk}{U_{\mathrm{chk}}}
\newcommand{\WR}{W_{\mathbb R}}
\newcommand{\HopfReal}{\mathrm{HopfReal}}
\newcommand{\HopfComplex}{\mathrm{HopfComplex}}
\newcommand{\EO}{E_O}
\newcommand{\bigO}{\mathcal{O}}
\newcommand{\indicator}[1]{\mathbf{1}\!\left\{#1\right\}}
\newcommand{\maybegraphic}[2][]{%
  \IfFileExists{#2}{\includegraphics[#1]{#2}}{%
    \fbox{\parbox[c][4.2cm][c]{0.92\linewidth}{\centering
    External circuit figure not found.\\[3pt]\texttt{\detokenize{#2}}}}%
  }%
}

\theoremstyle{definition}
\newtheorem{definition}{Definition}
\newtheorem{assumption}{Assumption}
\newtheorem{theorem}{Theorem}
\newtheorem{lemma}{Lemma}

\newcommand{\HKUSTGZ}{\affiliation{Thrust of Artificial Intelligence, Information Hub, The Hong Kong University of Science and Technology (Guangzhou), Guangzhou 511453, China}}
\newcommand{\QSC}{\affiliation{Quantum Science Center of Guangdong-Hong Kong-Macao Greater Bay Area, Shenzhen 518045, China}}

\hypersetup{
  pdftitle={Compass in the Mirror: Quantum Backpropagation with the Hopf Ansatz},
  pdfauthor={Ruge Lin, Ranyiliu Chen and Xin Wang},
  pdfsubject={Global-frame quantum backpropagation and layer-local gradients from the checkpoint method},
  pdfkeywords={quantum backpropagation, Hopf ansatz, variational quantum circuits, gradient estimation}
}

\begin{document}

\title{\texorpdfstring{Compass in the Mirror: \\
    Quantum Backpropagation with the Hopf Ansatz}{Compass in the Mirror: Quantum Backpropagation with the Hopf Ansatz}}

\author{Ruge Lin\,\orcidlink{0000-0002-9297-7622}}
\email[Corresponding author: ]{gogoko699@gmail.com}
\HKUSTGZ

\author{Ranyiliu Chen\,\orcidlink{0000-0001-6060-6192}}
\QSC

\begin{abstract}
  A variational quantum state may have many control angles, and its complete gradient contains one derivative for each. The task is to turn access to selected derivatives into a shared finite-shot record for the complete vector. The Hopf ansatz supplies the setting: it constructs any normalized pure state from nested two-way magnitude splits and, for complex states, leaf phases, with one directly programmed angle per split or phase. Earlier Hopf constructions prepare the state, recover coordinates from known amplitudes, and provide indexed derivative access. With a phase-calibrated controlled reflection, the inverse Hopf frame maps each normalized magnitude direction to a computational label, so every all-$X$ outcome contributes to every magnitude derivative. A separate signed leaf record supplies the complex phase block. A checkpoint variant reverses only the suffix below a selected depth and returns that layer's gradient block. For a complete chart with $M$ real-valued parameters, sufficient execution counts are $\bigO(\log\log M)$ for the global method and $\bigO(\log M\log\log M)$ for the complete checkpoint schedule at fixed additive coordinate accuracy and confidence. The global saving is record reuse, not a reduction of the $M$-entry output. Under the designated direct-angle schedules, the global method's assigned Hopf CNOT work and input/output-inclusive storage have at most logarithmic overhead relative to a matched scalar evaluation. The same Hopf geometry thus carries a state forward and the objective's complete coordinate response backward.
\end{abstract}

\maketitle

\section{Introduction}
\label{sec:introduction}

Variational quantum algorithms optimize a scalar objective by changing many control parameters in a state-preparation circuit. A selected derivative describes the response along one control, whereas a model update may require the complete vector. When separate measurement procedures are used for separate coordinates, much of the state preparation and objective evaluation is repeated. The task studied here is to use the structure of the state coordinates to make one quantum outcome informative about many derivatives at once.

The Hopf ansatz provides that structure~\cite{lin2026hopf}. It constructs computational-basis amplitudes through nested two-way splits: one angle distributes amplitude between two groups of basis states, and later splits refine each group. These splits form a balanced binary tree. Each magnitude coordinate is the physical rotation angle at its designated split, while each complex leaf phase is programmed directly. This one-coordinate--one-angle correspondence keeps the state representation, inverse coordinates, local geometry, and counted circuit controls in a common language rather than treating the state as the output of an arbitrary compiler.

Ref.~\cite{lin2026hopf} supplied direct-angle circuits for every normalized real or complex pure state, a classical inverse that recovers coordinates from known amplitudes, and a diagonal metric with normalized coordinate directions. Indexed tangent and signed-branch state preparations also made one selected derivative accessible. Together, these capabilities place state preparation, objective evaluation, and coordinate updates in one variational loop~\cite{cerezo2021variational}. The present construction extends that indexed interface to collective readout: it compiles the normalized magnitude directions into computationally addressed columns and uses their common objective response to produce a shared finite-shot record.

At a fixed Hopf parameter point, the requested output is the complete classical vector of raw coordinate derivatives of a scalar expectation. For $n$ qubits and $N=2^n$, the real and complex charts have $M=N-1$ and $M=2N-1$ real-valued parameters, respectively, so the output grows exponentially with $n$. The objective is assumed to admit a controlled reflection with known relative phase. This coherent access keeps the state and its reflected image in two branches whose interference exposes the transition moments needed for differentiation.

The addressed Hopf frame is a reversible basis change with the prepared state in its zero column and each normalized magnitude direction in a distinct labeled column. After the controlled reflection, the inverse frame transfers the relevant transition moments to those labels. Measuring the ancilla and system in $X$ produces one bit string whose parities supply signed samples for every magnitude direction. The samples decoded from one outcome are correlated, but each complete depth record has fixed Euclidean norm. A signed histogram and a single Walsh--Hadamard transform evaluate all the parity averages in a classical batch.

The complex construction retains this magnitude record after dressing the frame with the leaf phases. Phase derivatives are already localized in the computational basis, so a separate signed leaf record supplies the phase block without an inverse frame. For selected-depth access, the checkpoint construction reverses only the suffix below one chosen depth; the measured prefix identifies the derivative within that layer. The global frame is the principal collective construction. Checkpointing trades cross-depth record reuse for a shorter, layer-local reverse circuit.

Every magnitude-depth record has norm two, independently of its width. Vector concentration therefore controls all coordinates at a depth without a coordinatewise union bound; only the $n$ depth-level failure events need to be combined. For simultaneous additive coordinate accuracy $\varepsilon$ and confidence $1-\delta$, sufficient execution counts are
\begin{align*}
  S_{\mathrm{glob}} &= \bigO\!\left(\frac{1+\log(n/\delta)}{\varepsilon^2}\right),\\
  S_{\mathrm{chk}} &= \bigO\!\left( \frac{n\{1+\log(n/\delta)\}}{\varepsilon^2} \right).
\end{align*}
At fixed accuracy and confidence, these become $\bigO(\log\log M)$ and $\bigO(\log M\log\log M)$ independent executions. The saving comes from reusing each global record across depths; the classical output still contains $M$ entries.

The resource analysis counts controlled-reflection calls, assigned Hopf CNOT work, classical decoding, and storage separately. The matched comparison follows Ref.~\cite{abbas2023quantum}, with the logarithmic structured-circuit scale of Ref.~\cite{bowles2025backprop} as the benchmark. Under the designated direct-angle schedules, the global method has $\bigO(M\log M\log\log M)$ assigned Hopf work and meets this benchmark relative to a scalar evaluation using the same forward circuit and controlled reflection. Walsh decoding costs $\bigO(M\log M)$. These guarantees concern absolute raw-coordinate accuracy; relative or directional guarantees depend on a nonzero gradient signal, and metric-rescaled outputs have additional conditioning requirements.

Existing gradient methods address different combinations of access and output. Parameter shift and other low-depth rules commonly obtain selected coordinate derivatives from ordinary objective evaluations~\cite{mitarai2018quantum, schuld2019evaluating,banchi2021measuring,wierichs2022general, harrow2021lowdepth}. Classical simulator adjoints instead reverse a stored simulation~\cite{jones2020efficient}. Other quantum methods reuse measurements through coherent parameter registers, commutation, algebraic structure, generalized Hadamard tests, stabilizer descriptions, or classical-shadow measurements~\cite{verdon2018baqprop,bowles2025backprop, heidari2024lie,huang2020predicting,heidari2023shadow,li2024flexible, he2023singlecircuit,chinzei2025tradeoff}. Random-direction methods return stochastic search information whose variance and optimizer role differ from those of a uniformly accurate materialized vector~\cite{spall1992multivariate, coyle2026adaptive}. These approaches are not totally ordered. Hopf-QBP combines the stronger controlled-reflection interface with an addressed differential frame, allowing every global outcome to contribute to the complete raw magnitude-coordinate vector. Its reverse organization also connects to classical backpropagation, adjoint control, and checkpointed computation~\cite{rumelhart1986learning,baydin2018automatic, khaneja2005optimal,griewank2000revolve}.

The resulting interface may be useful when complete state-coordinate updates are needed and the objective admits an efficient controlled Pauli or reflection description. In variational scientific computing, such settings include deflation, simulation, process approximation, linear systems, and differential equations~\cite{higgott2019vqd,yuan2019vqs, endo2020generalprocesses,bravoprieto2023vqls,sato2021poisson,sarma2024pde}. Potential quantum-machine-learning uses include trainable embeddings, metric learning, kernel alignment, generator updates against fixed binary discriminators, and fidelity-trained pure-state models~\cite{lloyd2020quantumembeddings, sahin2024kernelalignment,dallairedemers2018qgan,niu2022eqgan, liu2020tomography}. Sensing and discrimination may use the same interface to optimize probes, sensor networks, and binary tests~\cite{koczor2020metrology, marciniak2022programmablesensors,subramanian2024hypothesistesting}. Checkpoint records may also support layer-local quantum-control updates when the control decomposition follows the Hopf-tree depth. In each setting, the data map, controlled-reflection compiler, aggregation rule, and optimizer remain application-specific; the result here supplies the coordinate-gradient interface rather than a complete application algorithm.

Section~\ref{sec:access-model} uses the two-qubit case to introduce the local differential pattern, objective access, and circuit contracts. Section~\ref{sec:reverse-geometry} builds the complete Hopf differential and its forward interfaces. Sections~\ref{sec:global-frame} and \ref{sec:complex-global} develop the real and complex global records, and Sec.~\ref{sec:checkpoint-adjoint} gives the selected-depth checkpoint construction. Section~\ref{sec:resource-analysis} combines their statistical and matched-resource guarantees. Protocols~\ref{prot:global-frame-gradient} and \ref{prot:checkpoint-gradient} summarize the two executable procedures. Appendix~\ref{app:four-qubit-example} provides a complete four-qubit gate-level companion, and Appendix~\ref{app:all-proofs} contains the proofs, accuracy results, and compiler derivations.

\section{The two-qubit clue and the reverse task}
\label{sec:access-model}

\subsection{One split, three directions}
\label{subsec:two-qubit-clue}

The two-qubit case contains the local split-and-complement pattern used throughout the paper. Take $n=2$ and $N=4$. The two depth-one subtree states are
\begin{align*}
  \ket{\sigma_2} &= \cos\theta_2\ket{00} + \sin\theta_2\ket{01},\\
  \ket{\sigma_3} &= \cos\theta_3\ket{10} + \sin\theta_3\ket{11},
\end{align*}
and the root combines them as
\begin{equation*}
  \ket{\psi^{\mathbb R}} = \cos\theta_1\ket{\sigma_2} + \sin\theta_1\ket{\sigma_3}.
\end{equation*}
Differentiating one split replaces its local state direction by the orthogonal partner,
\begin{align*}
  \ket{e_1^{\mathbb R}} &=-\sin\theta_1\ket{\sigma_2} +\cos\theta_1\ket{\sigma_3},\\
  \ket{e_2^{\mathbb R}} &=-\sin\theta_2\ket{00} +\cos\theta_2\ket{01},\\
  \ket{e_3^{\mathbb R}} &=-\sin\theta_3\ket{10} +\cos\theta_3\ket{11}.
\end{align*}
The derivatives are
\begin{align*}
  \partial_{\theta_1}\ket{\psi^{\mathbb R}} &=\ket{e_1^{\mathbb R}},\\
  \partial_{\theta_2}\ket{\psi^{\mathbb R}} &=\cos\theta_1\ket{e_2^{\mathbb R}},\\
  \partial_{\theta_3}\ket{\psi^{\mathbb R}} &=\sin\theta_1\ket{e_3^{\mathbb R}}.
\end{align*}
Every coordinate has one normalized direction, and a descendant derivative is weighted by the amplitude arriving at its node. The directions are mutually orthogonal: the two descendants occupy disjoint subtrees, and each local complement is orthogonal to the subtree state used by its ancestor. The balanced $n$-qubit construction repeats this split-and-complement rule at every internal node.

\subsection{What a reusable reverse record must provide}
\label{subsec:reverse-requirement}

For a normalized state family and a parameter-independent Hermitian observable, let
\begin{equation}
  \EO(\boldsymbol\theta) = \bra{\psi(\boldsymbol\theta)}O\ket{\psi(\boldsymbol\theta)}
  \label{eq:objective}
\end{equation}
be the scalar expectation objective. A real coordinate changes this objective through the transition moment
\begin{equation}
  \partial_{\theta_j}\EO = 2\Real\bra{\partial_{\theta_j}\psi}O\ket\psi.
  \label{eq:generic-gradient-transition}
\end{equation}
Normalization gives $\Real\braket{\psi|\partial_{\theta_j}\psi}=0$. Any component of the state differential parallel to $\ket\psi$ therefore contributes no change to the expectation; only the orthogonal component is needed.

A reusable reverse record needs four ingredients: orthogonal differential directions, their coordinate weights, a reversible map to computational columns, and a classical rule that reuses one outcome across a coordinate block. The Hopf tree supplies the weighted differential frame; controlled interference and classical decoding convert it into shared gradient records. Sections~\ref{sec:reverse-geometry}--\ref{sec:checkpoint-adjoint} construct the block estimators. Theorem~\ref{thm:hopf-quantum-backpropagation} gives their accuracy, work, and storage guarantees, and Appendix~\ref{app:generic-reverse} proves the common interference identities.

\subsection{Objective access and output target}
\label{subsec:qbp-target}

Fix a known Hopf state circuit, a parameter point, and a scalar expectation objective. The required output is the complete classical vector of raw Hopf-coordinate derivatives. The construction assumes that the objective is available coherently as a controlled reflection with known relative phase.

\begin{assumption}[Phase-calibrated controlled reflection]
  \label{ass:controlled-reflection}
  The observable is a known, parameter-independent Hermitian unitary,
  \begin{equation*}
    O=O^\dagger, \qquad O^2=I,
  \end{equation*}
  and the exact controlled operation
  \begin{equation*}
    \ctrl(O) = \ket0\!\bra0\otimes I + \ket1\!\bra1\otimes O
  \end{equation*}
  is available with a known relative phase between its branches, called the controlled-branch phase.
\end{assumption}

A classical description of $O$, or uncontrolled access to it, does not supply this interface. If the controlled branch implements $e^{\iu\gamma}O$, the $X$- and $Y$-readouts contain $\Real(e^{\iu\gamma}z)$ and $\Imag(e^{\iu\gamma}z)$. A known $\gamma$ is removed by a phase correction or calibrated equatorial readout. An unknown one changes the fixed decoder but does not affect the uncontrolled channel $O(\cdot)O^\dagger$.

Ref.~\cite{lin2026hopf} obtained the same moments from separately prepared state, tangent, and signed-branch states. Calibrated controlled access instead creates the coherent branch relation used by the shared records. This is a Hadamard-test instance~\cite{mitarai2019methodology}, with the calibration fixing the quadrature signs~\cite{araujo2014control}. A detector bit $b\in\{0,1\}$ represents eigenvalue $(-1)^b$.

The complete Hopf chart has $M$ real-valued parameters and returns
\begin{equation*}
  \nabla_{\boldsymbol\theta}\EO = \left( \partial_{\theta_1}\EO,\ldots,\partial_{\theta_M}\EO \right)^{\mathsf T} \in\mathbb R^M.
\end{equation*}
Throughout, $0<\varepsilon,\delta<1$.

\begin{definition}[Complete-gradient accuracy]
  \label{def:complete-gradient}
  An estimate $\widehat{\nabla_{\boldsymbol\theta}\EO}\in\mathbb R^M$ is $(\varepsilon,\delta)$ accurate if
  \begin{equation}
    \Prob\!\left[ \left\|\widehat{\nabla_{\boldsymbol\theta}\EO} -\nabla_{\boldsymbol\theta}\EO\right\|_\infty \leq\varepsilon \right] \geq1-\delta.
    \label{eq:complete-gradient-target}
  \end{equation}
\end{definition}

The condition gives simultaneous absolute coordinatewise accuracy for the complete classical output. Each shared record has a coordinate block as its mean.

\begin{definition}[State-coordinate quantum backpropagation]
  \label{def:state-coordinate-qbp}
  For a known differentiable state circuit $\boldsymbol\theta\mapsto\ket{\psi(\boldsymbol\theta)}$, a reverse protocol is called \emph{state-coordinate quantum backpropagation} if controlled objective interference is followed by a known resolution of the state differential. The resolution may be a reverse unitary or the identity when the relevant differential directions are already localized in the measurement basis. The protocol returns a vector-valued record $Z_{\mathcal I}$ satisfying
  \begin{equation}
    \E Z_{\mathcal I} = \left( 2\Real\bra{\partial_{\theta_j}\psi}O\ket\psi \right)_{j\in\mathcal I}
    \label{eq:qbp-block-definition}
  \end{equation}
  for a coordinate block $\mathcal I$. A complete protocol combines such shared blocks to produce the gradient in Definition~\ref{def:complete-gradient}.
\end{definition}

\medskip
\noindent\textbf{Matched resource benchmark.}
Following the complete-gradient versus function-evaluation comparison of
Ref.~\cite{abbas2023quantum}, we use the logarithmic overhead scale realized
in Ref.~\cite{bowles2025backprop}.  Under the direct-angle accounting of
Sec.~\ref{subsec:assigned-ledger}, the benchmark permits $\bigO(\log M)$
overhead in assigned Hopf CNOT work and input/output-inclusive storage.  The
matched scalar uses the same forward preparation and controlled reflection but
omits the reverse block; the gradient program materializes the complete vector.
The gradient protocols resolve $O\ket\psi$ globally in the addressed frame,
locally at a checkpoint, or directly in the leaf basis.

Each execution starts from a fresh product input and ends with destructive product-basis measurement. State preparation, reverse blocks, controlled reflection, readout basis changes, and classical processing are exact. Only sampling fluctuations enter the accuracy bounds. The system and branch ancilla use $n+1$ qubits, besides compiler and controlled-reflection workspace. All displayed execution and Hopf-work bounds are per reflection.

\medskip
\noindent\textbf{Reflection-sum objectives.}
Let $H=\sum_{\alpha=1}^{L}a_\alpha O_\alpha$ have fixed real coefficients
$a_\alpha$ and parameter-independent terms satisfying
$O_\alpha^\dagger=O_\alpha$ and $O_\alpha^2=I$.  Set
$\Lambda:=\sum_\alpha|a_\alpha|>0$ and assume calibrated controlled access
to each term.  Sample $\alpha$ with $p_\alpha=|a_\alpha|/\Lambda$ and return
$\widetilde Z_{\mathcal I}:=\Lambda\operatorname{sgn}(a_\alpha)
  Z_{\mathcal I}^{(\alpha)}$ for a record
$Z_{\mathcal I}^{(\alpha)}$ on block $\mathcal I$.  With
$E_H:=\bra\psi H\ket\psi$,
\begin{align*}
  \E\widetilde Z_{\mathcal I} &=\sum_\alpha a_\alpha\E Z_{\mathcal I}^{(\alpha)} =\left(\partial_{\theta_j}E_H\right)_{j\in\mathcal I},\\
  \|\widetilde Z_{\mathcal I}\|_2 &\leq\Lambda B \quad\text{whenever}\quad \|Z_{\mathcal I}^{(\alpha)}\|_2\leq B.
\end{align*}
The sufficient execution bounds therefore acquire a factor of $\Lambda^2$. The matched scalar uses the same term distribution and expected controlled-term charge $\sum_\alpha p_\alpha C(\ctrl(O_\alpha))$.

\subsection{Three circuit contracts}
\label{subsec:circuit-contracts}

The constructions use three levels of circuit equality:
\begin{align}
  U&=V, &&\text{full-unitary equality},\nonumber\\
  U\ket0^{\otimes n}&=V\ket0^{\otimes n}, &&\text{state-column equality},\nonumber\\
  UP&=VP, &&\text{active-interface equality}.
  \label{eq:three-circuit-contracts}
\end{align}
Here $P$ projects onto the interface subspace reached by the forward prefix. Full-unitary equality identifies the complete circuit action. State-column equality fixes the prepared state, which is sufficient when all subsequent operations are held fixed. Active-interface equality fixes exactly the columns reached at a checkpoint while leaving the unused completion free. The global-frame method uses designated columns of a complete frame; the checkpoint method uses the reached active interface.

This hierarchy separates forward correctness from compiler accounting. Any exact preparation with the same initialized state column may replace the forward circuit while the later controlled reflection, reverse block, and measurement remain fixed; the final record distribution is unchanged. The assigned-work statements below use the designated direct-angle Hopf schedules, which keep each chart coordinate directly tied to its physical control. A different state-equivalent compiler is evaluated with its own forward and reverse ledger. Appendix~\ref{app:four-qubit-example} shows all three logical contracts in one finite circuit family.

The two-qubit construction fixes the local differential pattern, and the access model identifies the transition moments to be recovered. Applying the split-and-complement rule throughout the tree supplies the weighted orthonormal frame and depth interfaces used by both reverse methods.

\section{The balanced Hopf differential and forward interfaces}
\label{sec:reverse-geometry}

\subsection{Repeating the local split across the tree}
\label{subsec:hopf-preliminaries}

Applying the local split-and-complement rule at every internal node produces a full differential frame and one computational marker per node. Let $n$ be the number of qubits and $N=2^n$. For $0\leq\ell<N$, the basis labels are
\begin{equation*}
  \ell=(q_n\cdots q_1)_2, \qquad \ket{b_\ell}=\ket{q_n\cdots q_1}.
\end{equation*}
A complete rooted binary tree has root $1$, children $2j$ and $2j+1$, internal nodes $1\leq j<N$, and leaves $N+\ell$. Internal nodes are indexed breadth first. The internal nodes at depth $d$ form
\begin{equation*}
  \mathcal L_d:=\{2^d,\ldots,2^{d+1}-1\}, \qquad 0\leq d<n.
\end{equation*}
Writing $j=2^d+r\in\mathcal L_d$, attach to that node the computational marker
\begin{equation*}
  \lambda(j):=(2r+1)2^{n-d-1}.
\end{equation*}
This is the leftmost leaf of the right child and will address the local complement created at node $j$.

To separate the coordinate blocks when needed, write
\begin{equation}
  \boldsymbol\theta_{\mathrm{mag}} :=(\theta_1,\ldots,\theta_{N-1}), \qquad \boldsymbol\theta_{\mathrm{ph}} :=(\theta_N,\ldots,\theta_{2N-1}).
  \label{eq:shared-hopf-coordinate-blocks}
\end{equation}
The real chart uses $\boldsymbol\theta=\boldsymbol\theta_{\mathrm{mag}}$, while the complex chart uses $\boldsymbol\theta=(\boldsymbol\theta_{\mathrm{mag}},\boldsymbol\theta_{\mathrm{ph}})$.

The real magnitude angles have ranges
\begin{align*}
  \theta_j&\in[0,\pi/2], &&1\leq j<N/2,\\
  \theta_j&\in[0,2\pi), &&N/2\leq j<N.
\end{align*}
The final layer carries arbitrary real signs. The inherited rotation convention is
\begin{equation}
  R_y(\theta)=e^{-\iu\theta Y} =
  \begin{pmatrix}
    \cos\theta & -\sin\theta \\
    \sin\theta & \cos\theta
  \end{pmatrix}.
  \label{eq:paper-Ry-convention}
\end{equation}
The native real state-preparation circuit is denoted $\HopfReal(\boldsymbol\theta_{\mathrm{mag}})$.

\subsection{Subtree states, complements, and weights}

For $0\leq\ell<N$, set the leaf states to
\begin{equation*}
  \ket{\sigma_{N+\ell}}:=\ket{b_\ell}.
\end{equation*}
Every internal node repeats the two-qubit split,
\begin{equation*}
  \ket{\sigma_j} = \cos\theta_j\ket{\sigma_{2j}} + \sin\theta_j\ket{\sigma_{2j+1}},
\end{equation*}
and the full real state is $\ket{\psi^{\mathbb R}(\boldsymbol\theta_{\mathrm{mag}})}:=\ket{\sigma_1}$. Within a construction at a fixed parameter point, the common argument $\boldsymbol\theta_{\mathrm{mag}}$ is suppressed.

The coefficient multiplying a subtree before its own split is the amplitude arriving at that node. Along the path from the root to an internal node, each left edge contributes a cosine and each right edge a sine. All such ancestor angles lie above the final sign-carrying layer and belong to $[0,\pi/2]$, so the incoming amplitude is nonnegative. Its square is denoted by $g^{\mathbb R}_{j,j}$. These quantities will also be the diagonal metric factors, and they obey
\begin{align}
  g^{\mathbb R}_{1,1}&=1,\nonumber\\
  g^{\mathbb R}_{2j,2j} &=g^{\mathbb R}_{j,j}\cos^2\theta_j,\nonumber\\
  g^{\mathbb R}_{2j+1,2j+1} &=g^{\mathbb R}_{j,j}\sin^2\theta_j,
  \label{eq:incoming-metric-recursion}
\end{align}
for internal children. Thus $\sqrt{g^{\mathbb R}_{j,j}(\boldsymbol\theta_{\mathrm{mag}})}$ is the incoming amplitude at node $j$.

The normalized direction orthogonal to the split state at node $j$ is its local complement,
\begin{equation}
  \ket{e_j^{\mathbb R}(\boldsymbol\theta_{\mathrm{mag}})} = -\sin\theta_j\ket{\sigma_{2j}} +\cos\theta_j\ket{\sigma_{2j+1}}.
  \label{eq:real-frame-column-recursion}
\end{equation}
The child-subtree states are orthonormal, so $\ket{e_j^{\mathbb R}}$ is normalized and orthogonal to $\ket{\sigma_j}$. Complements in disjoint subtrees have disjoint support. A descendant complement is also orthogonal to the descendant subtree state appearing in every ancestor combination. These support relations allow all local complements to coexist in one orthonormal basis.

The next result identifies the moving frame used by the reverse construction. It is the orthonormal basis whose first column is the state and whose remaining columns are the normalized coordinate directions, with completion columns at zero-weight coordinates.

\medskip
\noindent\textbf{Balanced Hopf differential guarantee.}
The set
\begin{equation}
  \left\{ \ket{\psi^{\mathbb R}}, \ket{e_1^{\mathbb R}},\ldots, \ket{e_{N-1}^{\mathbb R}} \right\}
  \label{eq:root-moving-frame}
\end{equation}
is an orthonormal basis. The incoming amplitude at node $j$ is also the coefficient of its coordinate derivative:
\begin{equation}
  \partial_{\theta_j} \ket{\psi^{\mathbb R}(\boldsymbol\theta_{\mathrm{mag}})} = \sqrt{g^{\mathbb R}_{j,j}(\boldsymbol\theta_{\mathrm{mag}})} \ket{e_j^{\mathbb R}}.
  \label{eq:hopf-derivative-factorization}
\end{equation}
Consequently,
\begin{equation}
  \Real\braket{ \partial_{\theta_i}\psi^{\mathbb R} | \partial_{\theta_j}\psi^{\mathbb R} } = \delta_{ij}g^{\mathbb R}_{j,j},
  \label{eq:hopf-diagonal-metric}
\end{equation}
so the squared incoming amplitudes are precisely the diagonal metric coefficients.

For $g^{\mathbb R}_{j,j}>0$, $\ket{e_j^{\mathbb R}}$ is the normalized coordinate tangent of Ref.~\cite{lin2026hopf}. When $g^{\mathbb R}_{j,j}=0$, the coordinate derivative vanishes and Eq.~\eqref{eq:real-frame-column-recursion} supplies the normalized completion column. For the objective in Eq.~\eqref{eq:objective},
\begin{equation}
  \partial_{\theta_j}\EO = 2\sqrt{g^{\mathbb R}_{j,j}(\boldsymbol\theta_{\mathrm{mag}})} \Real\bra{e_j^{\mathbb R}}O\ket{\psi^{\mathbb R}}.
  \label{eq:hopf-transition-gradient}
\end{equation}
Appendix~\ref{app:proof-real-geometry} proves the orthonormality, derivative, metric, and zero-weight statements. Appendix~\ref{app:four-qubit-metric} gives the four-qubit metric factors together with representative tangent-state formulas.

Each coordinate differential is a known weight times one moving-frame column. The reverse construction must make those columns accessible without preparing them separately. A depth-ordered forward completion first fixes the layer interfaces and the matched scalar baseline. The addressed frame then assigns a computational label to every normalized magnitude direction.

\subsection{Depth-ordered completion and the state-column contract}
\label{subsec:depth-completion-contracts}

A depth-ordered forward completion serves two later roles: it defines the designated scalar baseline and exposes the layer interfaces used by the checkpoint method. For $j=2^d+r\in\mathcal L_d$, define
\begin{align*}
  Q_{d,r} &:= \ket r\!\bra r_{q_n\cdots q_{n-d+1}} \otimes Y_{q_{n-d}}\otimes I,\\
  U_d &:= \prod_{r=0}^{2^d-1} \exp\!\left(-\iu\theta_{2^d+r}Q_{d,r}\right),\\
  \Uchk(\boldsymbol\theta_{\mathrm{mag}}) &:= U_{n-1}\cdots U_1U_0.
\end{align*}
The prefix projector is absent at $d=0$. Within one depth, the generators commute because their prefix projectors are orthogonal.

The depth completion and the native real circuit prepare the same initialized state:
\begin{equation}
  \HopfReal(\boldsymbol\theta_{\mathrm{mag}})\ket0^{\otimes n} = \Uchk(\boldsymbol\theta_{\mathrm{mag}})\ket0^{\otimes n} = \ket{\psi^{\mathbb R}}.
  \label{eq:depth-hopf-state}
\end{equation}
For correctness substitutions, write $U_{\mathrm{fwd}}^{\mathbb R}$ for any exact real preparation satisfying
\begin{equation}
  U_{\mathrm{fwd}}^{\mathbb R}\ket0^{\otimes n} =\ket{\psi^{\mathbb R}}.
  \label{eq:forward-preparation-contract}
\end{equation}
This is the real instance of the state-column principle in Sec.~\ref{subsec:circuit-contracts}. Native preparation may therefore be used in the record circuits, while $\Uchk$ supplies the checkpoint factorization and is the designated forward in the assigned Hopf CNOT ledger for the separated schedule. Both are direct-angle Hopf realizations: every $\theta_j$ remains the physical angle of its designated tree split. A further state-equivalent replacement preserves the record distributions but is assessed with its own compiler ledger. Appendix~\ref{app:proof-depth-completion} proves the common state column and the invariance of later measurements under its replacement; Appendix~\ref{app:four-qubit-depth-completion} displays the corresponding four-qubit depth circuit.

The remaining ingredient for the global method is a computationally addressed frame. Its inverse makes one reflected record readable across every magnitude depth.

\section{Global-frame method for the real Hopf chart}
\label{sec:global-frame}

\subsection{Construction of the addressed frame}

The addressed frame is a reversible circuit whose zero column is the prepared state and whose designated nonzero columns are the normalized magnitude directions. It places each local complement in the computational column marked by $\lambda(j)$. For $j=2^d+r\in\mathcal L_d$, define the leftmost leaf of its subtree by
\begin{equation*}
  \ell_0(j):=r2^{n-d}.
\end{equation*}
Let $G_j$ be the two-level rotation
\begin{align}
  G_j\ket{b_{\ell_0(j)}} &= \cos\theta_j\ket{b_{\ell_0(j)}} + \sin\theta_j\ket{b_{\lambda(j)}},
  \label{eq:frame-gate-anchor}\\
  G_j\ket{b_{\lambda(j)}} &= -\sin\theta_j\ket{b_{\ell_0(j)}} + \cos\theta_j\ket{b_{\lambda(j)}},
  \label{eq:frame-gate-marker}
\end{align}
and let it act as the identity elsewhere. Gates at one depth act on disjoint basis pairs. Applying parent layers before descendant layers gives
\begin{equation*}
  V_d:=\prod_{r=0}^{2^d-1}G_{2^d+r}, \qquad \WR(\boldsymbol\theta_{\mathrm{mag}}) :=V_{n-1}\cdots V_1V_0.
\end{equation*}

Read columnwise, the product has the following action. At node $j$, the gate $G_j$ creates the state and complement coefficient patterns on the two child anchors. Descendant gates then expand those anchors into the complete child-subtree states. The marker map assigns each resulting complement a distinct nonzero computational label.

\medskip
\noindent\textbf{Construction guarantee.}
The map $j\mapsto\lambda(j)$ is a bijection from the $N-1$ internal nodes to
the nonzero computational labels, and
\begin{equation}
  \WR(\boldsymbol\theta_{\mathrm{mag}})\ket0^{\otimes n} =\ket{\psi^{\mathbb R}}, \qquad \WR(\boldsymbol\theta_{\mathrm{mag}})\ket{b_{\lambda(j)}} =\ket{e_j^{\mathbb R}}.
  \label{eq:frame-column-identities}
\end{equation}
Every normalized magnitude direction therefore has a reversible computational address. Appendix~\ref{app:proof-real-frame-columns} proves the marker bijection and column induction. Appendix~\ref{app:four-qubit-real-frame} lists the $n=4$ addressed gates.

\subsection{Why the inverse frame is readable}

With the differential columns computationally addressed, controlled interference turns their transition amplitudes into jointly measurable parity moments. Whenever $W\ket0^{\otimes n}=\ket\psi$, applying the inverse frame after the controlled reflection gives
\begin{equation}
  \ket{\Omega_W} = \frac{ \ket0_a\ket0^{\otimes n} + \ket1_aW^\dagger O\ket\psi }{\sqrt2}.
  \label{eq:generic-frame-interference-state}
\end{equation}
For a basis label $k=(k_n\cdots k_1)_2$, define
\begin{equation*}
  X(k):=\bigotimes_{q=1}^nX_q^{k_q}, \qquad k\cdot\boldsymbol y :=\sum_{q=1}^nk_qy_q\pmod2.
\end{equation*}
An all-$X$ outcome $(b,\boldsymbol y)$ then has parity expectation
\begin{equation}
  \E\!\left[(-1)^{b+k\cdot\boldsymbol y}\right] = \Real\bra{b_k}W^\dagger O\ket\psi.
  \label{eq:frame-parity-identity}
\end{equation}
Thus one measured bit string supplies a $\{\pm1\}$ sample for every addressed frame column. Repeated outcomes estimate all corresponding transition moments. Setting $W=\WR$ and $k=\lambda(j)$ recovers Eq.~\eqref{eq:hopf-transition-gradient}. Appendix~\ref{app:generic-reverse} proves the generic identity, and Appendix~\ref{app:proof-real-readout} proves its Hopf-frame specialization.

\medskip
\noindent\textbf{What controlled interference contributes.}
Write
\[
  W^\dagger O\ket\psi = \sum_{\ell=0}^{N-1}v_\ell\ket{b_\ell}.
\]
Equation~\eqref{eq:frame-parity-identity} reads $\Real v_k$ linearly from the coherent two-branch state. For comparison, let $p_{\mathrm{unc}}(\boldsymbol y)$ be the all-$X$ probability obtained by applying $O$, then $W^\dagger$, and measuring only the system. Its Walsh coefficient is
\[
  \sum_{\boldsymbol y} (-1)^{k\cdot\boldsymbol y}p_{\mathrm{unc}}(\boldsymbol y) = \sum_{\ell=0}^{N-1}v_\ell v_{\ell\oplus k}^{*},
\]
where $\ell\oplus k$ is bitwise addition modulo two. System probabilities therefore encode amplitude autocorrelations. The coherent branch ancilla supplies the relative phase that turns the same parity analysis into the linear moment required by the gradient.

\subsection{Circuit template, outcomes, and decoder}

\medskip
\noindent\textbf{One circuit template.}
At a fixed parameter point, the native-forward real implementation prepares
$\ket+_a\ket0^{\otimes n}$ and applies
\begin{equation}
  \HopfReal(\boldsymbol\theta_{\mathrm{mag}}) \longrightarrow \ctrl(O) \longrightarrow \WR^\dagger(\boldsymbol\theta_{\mathrm{mag}}),
  \label{eq:real-global-circuit-order}
\end{equation}
followed by $X$-basis measurement of the ancilla and all $n$ system qubits. This gate sequence and readout define the circuit template. One fresh execution returns the joint $(n+1)$-bit outcome
\begin{equation*}
  (b_m,\boldsymbol y), \qquad b_m\in\{0,1\}, \qquad \boldsymbol y=(y_1,\ldots,y_n)\in\{0,1\}^n.
\end{equation*}
The state-column contract in Sec.~\ref{subsec:circuit-contracts} permits any exact $U_{\mathrm{fwd}}^{\mathbb R}$ satisfying Eq.~\eqref{eq:forward-preparation-contract} to serve as the forward block for correctness. The native circuit displays the template; the assigned-work comparison uses the designated direct-angle forward of Sec.~\ref{subsec:assigned-ledger}.

\medskip
\noindent\textbf{Single-execution record and direct average.}
For every internal node, the same joint outcome defines
\begin{equation}
  Z_j^{\mathbb R}(b_m,\boldsymbol y) := 2\sqrt{g^{\mathbb R}_{j,j}} (-1)^{b_m+\lambda(j)\cdot\boldsymbol y}.
  \label{eq:real-hopf-estimator}
\end{equation}
Given $S$ independent outcomes $\{(b_{m,t},\boldsymbol y_t)\}_{t=1}^{S}$, the empirical coordinate gradient is
\begin{align}
  \widehat{\partial_{\theta_j}\EO} &= \frac1S\sum_{t=1}^{S} Z_j^{\mathbb R}(b_{m,t},\boldsymbol y_t)\nonumber\\
  &= \frac{2\sqrt{g^{\mathbb R}_{j,j}}}{S} \sum_{t=1}^{S} (-1)^{b_{m,t}+\lambda(j)\cdot\boldsymbol y_t}.
  \label{eq:real-hopf-direct-average}
\end{align}
The same $S$ outcomes enter Eq.~\eqref{eq:real-hopf-direct-average} for every $j$; only the predetermined marker $\lambda(j)$ and known metric factor change. Records from one execution are correlated, but distinct executions are independent.

For $n=2$, $\lambda(1)=(10)_2$, $\lambda(2)=(01)_2$, and $\lambda(3)=(11)_2$; one outcome $(b_m,y_1,y_2)$ supplies
\begin{align*}
  Z_1^{\mathbb R} &=2(-1)^{b_m+y_2},\\
  Z_2^{\mathbb R} &=2\cos\theta_1(-1)^{b_m+y_1},\\
  Z_3^{\mathbb R} &=2\sin\theta_1(-1)^{b_m+y_1+y_2}.
\end{align*}
The same $S$ recorded outcomes enter all three averages.

\medskip
\noindent\textbf{Signed histogram and Walsh decoder.}
The classical decoder first accumulates the joint counts
\begin{equation}
  C_{b,\boldsymbol y} := \sum_{t=1}^{S} \indicator{b_{m,t}=b,\ \boldsymbol y_t=\boldsymbol y}
  \label{eq:real-global-joint-counts}
\end{equation}
The counts obey $\sum_{\boldsymbol y}(C_{0,\boldsymbol y}+C_{1,\boldsymbol y})=S$. Define their signed histogram by
\begin{equation*}
  h(\boldsymbol y) := C_{0,\boldsymbol y}-C_{1,\boldsymbol y} = \sum_{t:\boldsymbol y_t=\boldsymbol y}(-1)^{b_{m,t}}.
\end{equation*}
A fast Walsh--Hadamard transform~\cite{fino1976walsh} then evaluates all parity moments together:
\begin{equation}
  \widehat m_k = \frac1S\sum_{\boldsymbol y} h(\boldsymbol y)(-1)^{k\cdot\boldsymbol y}.
  \label{eq:fwht-moments}
\end{equation}
The $k=0$ entry is the marginal ancilla imbalance and estimates only $\EO$. The other system-resolved parity entries retain the frame components needed for the complete gradient. After a tree pass computes the metric factors,
\begin{equation}
  \widehat{\partial_{\theta_j}\EO} = 2\sqrt{g^{\mathbb R}_{j,j}}\widehat m_{\lambda(j)}.
  \label{eq:weight-rescaling}
\end{equation}
The transform is the batched implementation of the direct averages in Eq.~\eqref{eq:real-hopf-direct-average}; it requires no additional quantum execution. The complete classical decoder uses $\bigO(S+N\log N)$ additions and subtractions and stores $\bigO(N)$ auxiliary histogram values.

\subsection{Global real guarantee and cost at a glance}

\medskip
\noindent\textbf{Guarantee.}
For every internal node,
\begin{equation}
  \E Z_j^{\mathbb R}=\partial_{\theta_j}\EO, \qquad |Z_j^{\mathbb R}|\leq2\sqrt{g^{\mathbb R}_{j,j}}\leq2.
  \label{eq:real-hopf-estimator-properties}
\end{equation}
When $g^{\mathbb R}_{j,j}=0$, the record is identically zero. At depth $d$, collect the magnitude coordinates as
\begin{align*}
  Z_d^{\mathrm{fr}} &:= \left(Z_{2^d+r}^{\mathbb R}\right)_{r=0}^{2^d-1},\\
  \mu_d^{\mathrm{mag}} &:= \left(\partial_{\theta_{2^d+r}}\EO\right)_{r=0}^{2^d-1}.
\end{align*}
Every outcome satisfies
\begin{equation}
  \E Z_d^{\mathrm{fr}}=\mu_d^{\mathrm{mag}}, \qquad \|Z_d^{\mathrm{fr}}\|_2=2.
  \label{eq:global-depth-fixed-norm}
\end{equation}
Let $S_{\mathrm{mag}}^{\mathbb R}$ be the number of independent records in the real magnitude stream. The global-frame method repeats one template, and the same records serve every depth. The fixed-norm concentration result in Sec.~\ref{subsec:statistical-tools} gives the sufficient allocation in Eq.~\eqref{eq:complete-chart-shot-allocations}; one common count covers all coordinates and depths. Appendix~\ref{app:proof-real-readout} proves the parity sign, unbiasedness, zero-weight behavior, depth normalization, and Walsh decoder. These identities supply the real case of Theorem~\ref{thm:hopf-quantum-backpropagation}(A).

\medskip
\noindent\textbf{Cost at a glance.}
For the complete real chart, $M=N-1$.  At one parameter point, the template in
Eq.~\eqref{eq:real-global-circuit-order} is executed
$S_{\mathrm{mag}}^{\mathbb R}$ times.  Each execution returns one joint
$(n+1)$-bit outcome contributing to all $N-1$ coordinates.  Outcomes stream
into the $N$ signed-histogram bins, avoiding storage of the full raw record
list.  For the native and designated Hopf preparations, the forward block and
addressed inverse frame each require $\bigO(nN)$ assigned Hopf CNOT work per
execution~\cite{lin2026hopf}.  The application-dependent controlled reflection
and readout are reported separately.  Classical decoding costs
$\bigO(S_{\mathrm{mag}}^{\mathbb R}+N\log N)$ operations, uses $\bigO(N)$
auxiliary storage, and returns $N-1$ entries.  The live quantum register
contains the system, branch ancilla, and any compiler workspace.
Sec.~\ref{subsec:assigned-ledger} compares this assigned Hopf CNOT work with
the matched scalar program used in the theorem.

The real construction establishes the shared magnitude mechanism: one global-frame template is repeated, and every outcome is reused across all tree depths. The complex chart preserves this mechanism after phase dressing and adds a separate leaf-local record to complete the phase-coordinate block.

\section{Global-frame method for the complex Hopf chart}
\label{sec:complex-global}

\subsection{Complex chart and phase-dressed frame}

Using the coordinate blocks in Eq.~\eqref{eq:shared-hopf-coordinate-blocks}, write $\boldsymbol\theta=(\boldsymbol\theta_{\mathrm{mag}}, \boldsymbol\theta_{\mathrm{ph}})$. All magnitude angles now lie in $[0,\pi/2]$, and
\begin{equation*}
  \theta_{N+\ell}\in[0,2\pi), \qquad 0\leq\ell<N.
\end{equation*}
Define the nonnegative magnitude state and diagonal phase unitary by
\begin{align*}
  \ket{\psi_{\mathrm{mag}}(\boldsymbol\theta_{\mathrm{mag}})} &:= \Uchk(\boldsymbol\theta_{\mathrm{mag}})\ket0^{\otimes n},\\
  \Dph(\boldsymbol\theta_{\mathrm{ph}}) &:= \sum_{\ell=0}^{N-1} e^{\iu\theta_{N+\ell}}\ket{b_\ell}\!\bra{b_\ell},\\
  \ket{\psi^{\mathbb C}(\boldsymbol\theta)} &:= \Dph(\boldsymbol\theta_{\mathrm{ph}}) \ket{\psi_{\mathrm{mag}}(\boldsymbol\theta_{\mathrm{mag}})}.
\end{align*}
Write $x_\ell(\boldsymbol\theta):= \braket{b_\ell|\psi^{\mathbb C}(\boldsymbol\theta)}$. The native complex preparation satisfies
\begin{align}
  \HopfComplex(\boldsymbol\theta)\ket0^{\otimes n} &= \Dph(\boldsymbol\theta_{\mathrm{ph}}) \Uchk(\boldsymbol\theta_{\mathrm{mag}})\ket0^{\otimes n}\nonumber\\
  &=\ket{\psi^{\mathbb C}(\boldsymbol\theta)}.
  \label{eq:designated-complex-preparation}
\end{align}
For correctness substitutions, write $U_{\mathrm{fwd}}^{\mathbb C}$ for any exact complex preparation satisfying
\begin{equation}
  U_{\mathrm{fwd}}^{\mathbb C}\ket0^{\otimes n} =\ket{\psi^{\mathbb C}(\boldsymbol\theta)}.
  \label{eq:complex-forward-preparation-contract}
\end{equation}
This is the complex instance of the state-column principle in Sec.~\ref{subsec:circuit-contracts}. The designated assigned-work schedule is the direct-angle preparation $\Dph\Uchk$; the native $\HopfComplex$ circuit is another direct-angle state-column completion. Appendix~\ref{app:four-qubit-native} recalls the native gate-level realization, including the complex two-leaf completion.

Magnitude differentiation keeps the phase layer fixed. The resulting normalized direction is the phase-dressed real complement,
\begin{equation*}
  \ket{e_j^{\mathbb C}(\boldsymbol\theta)} := \Dph(\boldsymbol\theta_{\mathrm{ph}}) \ket{e_j^{\mathbb R}(\boldsymbol\theta_{\mathrm{mag}})}.
\end{equation*}
The complex magnitude metric equals the real tree metric,
\begin{equation}
  g^{\mathbb C}_{j,j}(\boldsymbol\theta) = g^{\mathbb R}_{j,j}(\boldsymbol\theta_{\mathrm{mag}}), \qquad 1\leq j<N.
  \label{eq:complex-real-magnitude-metric}
\end{equation}
Therefore,
\begin{align}
  \partial_{\theta_j}\ket{\psi^{\mathbb C}(\boldsymbol\theta)} &= \sqrt{g^{\mathbb C}_{j,j}(\boldsymbol\theta)} \ket{e_j^{\mathbb C}(\boldsymbol\theta)},\nonumber\\
  \partial_{\theta_j}\EO &= 2\sqrt{g^{\mathbb C}_{j,j}(\boldsymbol\theta)} \Real\bra{e_j^{\mathbb C}}O\ket{\psi^{\mathbb C}}, \qquad 1\leq j<N.
  \label{eq:complex-magnitude-derivative}
\end{align}
The complete complex magnitude frame is
\begin{align*}
  W_{\mathbb C}(\boldsymbol\theta) &:= \Dph(\boldsymbol\theta_{\mathrm{ph}}) \WR(\boldsymbol\theta_{\mathrm{mag}}),\\
  W_{\mathbb C}^{\dagger}(\boldsymbol\theta) &= \WR^\dagger(\boldsymbol\theta_{\mathrm{mag}}) \Dph^\dagger(\boldsymbol\theta_{\mathrm{ph}}).
\end{align*}
Its columns satisfy
\begin{equation*}
  W_{\mathbb C}(\boldsymbol\theta)\ket0^{\otimes n} =\ket{\psi^{\mathbb C}(\boldsymbol\theta)}, \qquad W_{\mathbb C}(\boldsymbol\theta)\ket{b_{\lambda(j)}} =\ket{e_j^{\mathbb C}(\boldsymbol\theta)}.
\end{equation*}
Common parameter arguments may now be suppressed. The displayed definition gives the separated inverse frame. A compiler may absorb phase dressing into addressed gates; Appendix~\ref{app:four-qubit-complex-frame} gives the four-qubit realization.

\noindent\textbf{Phase-metric convention.}
For $p_\ell:=|x_\ell|^2$ and
$\boldsymbol p=(p_0,\ldots,p_{N-1})^{\mathsf T}$, the complex chart uses the
ambient round-sphere metric of Ref.~\cite{lin2026hopf},
\begin{equation*}
  G_{\mathrm{ph}}^{\mathrm{sph}}=\operatorname{diag}(\boldsymbol p).
\end{equation*}
Quotienting global phase would instead give the unused projective comparison $G_{\mathrm{ph}}^{\mathrm{FS}}=\operatorname{diag}(\boldsymbol p) -\boldsymbol p\boldsymbol p^{\mathsf T}$. Uniform phase is not metric-null in the adopted geometry, although expectation objectives are invariant along it.

The complete complex gradient uses two templates. The magnitude template reuses each joint all-$X$ outcome across the $N-1$ magnitude coordinates. The phase template places each ancilla-$Y$/system-$Z$ outcome in a signed one-hot histogram for the $N$ leaf phases. They share state preparation and $\ctrl(O)$ but run as separate circuit families with different reverse blocks and readouts.

\subsection{Complex magnitude circuit and parity reuse}

\medskip
\noindent\textbf{One magnitude circuit template.}
At a fixed parameter point, prepare $\ket+_a\ket0^{\otimes n}$ and apply
\begin{equation*}
  U_{\mathrm{fwd}}^{\mathbb C} \longrightarrow \ctrl(O) \longrightarrow W_{\mathbb C}^{\dagger},
\end{equation*}
followed by $X$-basis measurement of every qubit. The accounting of assigned Hopf CNOT work for the designated separated implementation uses the block order
\begin{equation*}
  \Uchk \longrightarrow \Dph \longrightarrow \ctrl(O) \longrightarrow \Dph^\dagger \longrightarrow \WR^\dagger.
\end{equation*}
The state-column identity in Eq.~\eqref{eq:designated-complex-preparation} also permits the native $\HopfComplex$ circuit as the forward block. Both forward choices belong to the same complex-magnitude measurement family and retain direct coordinate-to-angle control; the theorem's finite constants refer to the separated schedule.

One execution returns the joint $(n+1)$-bit outcome $(b_m,\boldsymbol y)$ used by the real global-frame method. For every magnitude coordinate, define
\begin{equation}
  Z_j^{\mathbb C,\mathrm{mag}}(b_m,\boldsymbol y) := 2\sqrt{g^{\mathbb C}_{j,j}} (-1)^{b_m+\lambda(j)\cdot\boldsymbol y}, \qquad 1\leq j<N.
  \label{eq:complex-magnitude-estimator}
\end{equation}
Let $S_{\mathrm{mag}}^{\mathbb C}$ be the number of independent magnitude executions. Their empirical average is
\begin{equation}
  \widehat{\partial_{\theta_j}\EO} = \frac{1}{S_{\mathrm{mag}}^{\mathbb C}} \sum_{t=1}^{S_{\mathrm{mag}}^{\mathbb C}} Z_j^{\mathbb C,\mathrm{mag}}(b_{m,t},\boldsymbol y_t).
  \label{eq:complex-magnitude-direct-average}
\end{equation}
The same outcomes serve every $j$; only the marker and known metric factor change. Because $g^{\mathbb C}_{j,j}=g^{\mathbb R}_{j,j}$, the signed histogram and Walsh transform in Eqs.~\eqref{eq:real-global-joint-counts}--\eqref{eq:weight-rescaling} apply without modification.

\subsection{Direct phase circuit and signed leaf histogram}

Each phase derivative is already supported on one computational-basis leaf and needs no inverse frame:
\begin{equation}
  \partial_{\theta_{N+\ell}}\ket{\psi^{\mathbb C}(\boldsymbol\theta)} = \iu x_\ell(\boldsymbol\theta)\ket{b_\ell}, \qquad 0\leq\ell<N.
  \label{eq:complex-phase-raw-tangent}
\end{equation}
Computational-basis readout is therefore the identity differential resolution in Definition~\ref{def:state-coordinate-qbp}. The corresponding objective derivative is
\begin{equation*}
  \partial_{\theta_{N+\ell}}\EO = 2\Imag\bra{\psi^{\mathbb C}}\Pi_\ell O\ket{\psi^{\mathbb C}}, \qquad \Pi_\ell:=\ket{b_\ell}\!\bra{b_\ell}.
\end{equation*}

\medskip
\noindent\textbf{One phase circuit template.}
Prepare $\ket+_a\ket0^{\otimes n}$ with $\HopfComplex$, apply $\ctrl(O)$,
then measure the ancilla in $Y$ and the system in $Z$.  The state-column
contract permits any exact $U_{\mathrm{fwd}}^{\mathbb C}$ satisfying
Eq.~\eqref{eq:complex-forward-preparation-contract} for correctness; the
assigned-work comparison retains the designated direct-angle schedule.  One
execution returns
$(b_p,\widehat\ell)$, where $b_p\in\{0,1\}$ and
$\widehat\ell=(z_n\cdots z_1)_2\in\{0,\ldots,N-1\}$.

\medskip
\noindent\textbf{Single-execution record and direct average.}
For each leaf,
\begin{equation}
  Z_{N+\ell}^{\mathbb C,\mathrm{ph}}(b_p,\widehat\ell) := 2(-1)^{b_p}\indicator{\widehat\ell=\ell}.
  \label{eq:complex-phase-estimator}
\end{equation}
Here and below, $\mathbf e_r$ denotes a classical standard basis vector; it is distinct from the normalized quantum tangent kets $\ket{e_j^{\mathbb R}}$ and $\ket{e_j^{\mathbb C}}$. Equivalently, with $\mathbf e_{\widehat\ell}$ the standard basis vector of $\mathbb R^N$, the complete phase record is
\begin{equation}
  Z^{\mathbb C,\mathrm{ph}}(b_p,\widehat\ell) = 2(-1)^{b_p}\mathbf e_{\widehat\ell}, \qquad \|Z^{\mathbb C,\mathrm{ph}}\|_2=2.
  \label{eq:complex-phase-onehot-vector}
\end{equation}
Let $S_{\mathrm{ph}}^{\mathbb C}$ be the number of independent phase outcomes. Their empirical average is
\begin{align*}
  \widehat{\partial_{\theta_{N+\ell}}\EO} &= \frac{2}{S_{\mathrm{ph}}^{\mathbb C}} \sum_{t=1}^{S_{\mathrm{ph}}^{\mathbb C}} (-1)^{b_{p,t}}\indicator{\widehat\ell_t=\ell}\\
  &= \frac{2}{S_{\mathrm{ph}}^{\mathbb C}} \left(C_{0,\ell}^{\mathrm{ph}}-C_{1,\ell}^{\mathrm{ph}}\right),
\end{align*}
where $C_{b,\ell}^{\mathrm{ph}} :=\#\{t:(b_{p,t},\widehat\ell_t)=(b,\ell)\}$. Each outcome updates one leaf bin, and the complete signed histogram returns all $N$ phase derivatives. Its classical decoding uses $\bigO(S_{\mathrm{ph}}^{\mathbb C}+N)$ operations and $\bigO(N)$ slots.

The record naturally retains all $N$ phase coordinates in the order of Ref.~\cite{lin2026hopf}. A uniform leaf-phase shift leaves the expectation objective unchanged, so the exact phase-gradient block is zero-sum. At a zero-probability leaf, the corresponding phase differential and exact derivative vanish; the direct estimator remains defined without division by the leaf amplitude.

Protocol~\ref{prot:global-frame-gradient} summarizes the single-execution estimators. Each magnitude outcome serves all depths, and the phase stream completes the complex gradient.

\begin{protocolcard}{Global-frame method: \\one execution serves all depths}
  \label{prot:global-frame-gradient}
  \protocolstep{ONE MAGNITUDE EXECUTION}
  Prepare the real or complex Hopf state, apply $\ctrl(O)$, and apply the
  chart-specific inverse frame $\WR^\dagger$ or $W_{\mathbb C}^\dagger$.
  Measure the ancilla and system in $X$, obtaining $(b_m,\boldsymbol y)$.

  \protocolstep{WHAT THAT EXECUTION MEANS}
  For every internal node $j$, the same outcome enters the real or complex
  record in Eq.~\eqref{eq:real-hopf-estimator} or
  Eq.~\eqref{eq:complex-magnitude-estimator}.  Its ancilla and Walsh signs,
  together with the known metric factor, form one simultaneous sample of the
  complete magnitude-gradient block.

  \protocolstep{AFTER MANY EXECUTIONS}
  The signed histogram of $\boldsymbol y$ contains all required parity moments.
  One fast Walsh--Hadamard transform evaluates them together; metric rescaling
  then returns the complete magnitude block.

  \protocolstep{COMPLEX PHASES}
  For the complex chart, a second family of executions measures the ancilla in
  $Y$ and the system in $Z$.  Each outcome names one leaf and contributes a
  signed one-hot sample to the phase-gradient block.

  \protocolstep{RETURN}
  The real chart returns the magnitude block.  The complex chart concatenates
  the magnitude and phase blocks.
\end{protocolcard}

\subsection{Complex guarantee and cost at a glance}

\medskip
\noindent\textbf{Guarantee.}
The magnitude and phase records satisfy
\begin{align}
  \E Z_j^{\mathbb C,\mathrm{mag}} &=\partial_{\theta_j}\EO, & |Z_j^{\mathbb C,\mathrm{mag}}| &\leq2\sqrt{g^{\mathbb C}_{j,j}}\leq2,
  \label{eq:complex-magnitude-estimator-bound}\\
  \E Z_{N+\ell}^{\mathbb C,\mathrm{ph}} &=\partial_{\theta_{N+\ell}}\EO.
  \label{eq:complex-phase-estimator-properties}
\end{align}
Every complex magnitude depth has the norm-two property in Eq.~\eqref{eq:global-depth-fixed-norm}, and every phase record has norm two by Eq.~\eqref{eq:complex-phase-onehot-vector}. With the allocations $S_{\mathrm{mag}}^{\mathbb C}$ and $S_{\mathrm{ph}}^{\mathbb C}$ in Eq.~\eqref{eq:complete-chart-shot-allocations}, the all-$X$ magnitude and separate $Y/Z$ phase templates recover the complete complex coordinate vector. Appendix~\ref{app:proof-complex} proves phase dressing, magnitude parity, the phase-quadrature sign, the one-hot identity, and the compiler global-phase identity. The two records supply the complex case of Theorem~\ref{thm:hopf-quantum-backpropagation}(A).

\medskip
\noindent\textbf{Cost at a glance.}
For the complete complex Hopf chart, $M=2N-1$, and the gradient uses separate
magnitude and phase templates.  Their total global execution
count is
$S_{\mathrm{glob}}^{\mathbb C}:=S_{\mathrm{mag}}^{\mathbb C}+S_{\mathrm{ph}}^{\mathbb C}$.
Assigned Hopf CNOT work sums the selected forward and reverse block charges over
these executions; the controlled reflection is reported as the matched
application component.  Both streams accumulate in length-$N$ histograms.
Classical decoding costs
$\bigO(S_{\mathrm{glob}}^{\mathbb C}+N\log N)$, auxiliary storage is
$\bigO(N)$, and the output has $2N-1$ entries.  At fixed accuracy and confidence,
$S_{\mathrm{glob}}^{\mathbb C}=\bigO(\log\log M)$.

The global-frame method now returns the complete real or complex Hopf gradient, with one magnitude stream shared across every depth. Checkpointing uses the same differential structure at a different resolution: it replaces the complete inverse frame by a layer-local reverse suffix when only selected depths are needed.

\section{Checkpoint method: layer-local reverse gradients}
\label{sec:checkpoint-adjoint}

\subsection{From a complete frame to one active depth}

The checkpoint method returns the reflected response to one selected tree layer. The forward prefix creates the state at depth $d$; after the controlled reflection, the inverse suffix returns the reflected branch to that layer. The measured prefix then identifies the active coordinate. This is the checkpointed adjoint mechanism.

For depth $d$, factor the depth completion as
\begin{align*}
  A_d&:=U_d\cdots U_1U_0,\\
  B_d&:=U_{n-1}\cdots U_{d+1}, \qquad \Uchk=B_dA_d,
\end{align*}
and define the checkpoint state
\begin{equation*}
  \ket{\alpha_d}:=A_d\ket0^{\otimes n}.
\end{equation*}
The real reflected branch after suffix reversal is
\begin{equation*}
  \ket{\beta_d^{\mathbb R}} :=B_d^\dagger O\ket{\psi^{\mathbb R}}.
\end{equation*}
For $j=2^d+r$, differentiation at the selected depth gives
\begin{equation}
  \partial_{\theta_j}\EO = 2\Imag\bra{\beta_d^{\mathbb R}}Q_{d,r}\ket{\alpha_d}.
  \label{eq:checkpoint-adjoint-identity}
\end{equation}
Thus $B_d^\dagger$ is the required real reverse block.

\subsection{Real checkpoint template and decoder}

\medskip
\noindent\textbf{One depth-specific circuit template.}
For a selected depth $d$, prepare $\ket+_a\ket0^{\otimes n}$ and apply
\begin{equation*}
  \HopfReal(\boldsymbol\theta_{\mathrm{mag}}) \longrightarrow \ctrl(O) \longrightarrow B_d^\dagger(\boldsymbol\theta_{\mathrm{mag}}).
\end{equation*}
The state-column contract permits any exact $U_{\mathrm{fwd}}^{\mathbb R}$ satisfying Eq.~\eqref{eq:forward-preparation-contract} for correctness. The assigned-work comparison fixes $\Uchk$ as the designated direct-angle forward for this checkpoint program and its matched scalar baseline. After suffix reversal, the two branches are
\begin{equation}
  \ket{\Omega_d^{\mathbb R}} = \frac{ \ket0_a\ket{\alpha_d} + \ket1_a\ket{\beta_d^{\mathbb R}} }{\sqrt2}.
  \label{eq:checkpoint-interference-state}
\end{equation}
The $Y$-quadrature identity in Appendix~\ref{app:generic-reverse} converts their interference into the transported derivative in Eq.~\eqref{eq:checkpoint-adjoint-identity}.

\medskip
\noindent\textbf{One joint outcome.}
Measure the branch ancilla in $Y$, the $d$ prefix qubits in $Z$, and the target
$q_{n-d}$ in $Y$.  The record uses the ancilla, prefix, and target; the lower
suffix need not be measured.  One execution returns
$(b_c,b_t,\widehat r)$, where
$\widehat r=(z_n\cdots z_{n-d+1})_2\in\{0,\ldots,2^d-1\}$,
with $\widehat r=0$ at depth zero.

\medskip
\noindent\textbf{Single-execution record and direct average.}
The outcome defines
\begin{equation}
  Z_d^{\mathrm{chk}}(b_c,b_t,\widehat r) := -2(-1)^{b_c+b_t}\mathbf e_{\widehat r} \in\mathbb R^{2^d},
  \label{eq:checkpoint-onehot-vector}
\end{equation}
where $\mathbf e_{\widehat r}$ is the standard basis vector indexed by the measured prefix. Let $S_d$ be the number of records collected at depth $d$. Their average gives $\widehat\mu_d^{\mathrm{mag}}=S_d^{-1}\sum_{s=1}^{S_d}Z_{d,s}^{\mathrm{chk}}$. Componentwise,
\begin{equation}
  \widehat{\partial_{\theta_{2^d+r}}\EO} = -\frac{2}{S_d} \sum_{s=1}^{S_d} (-1)^{b_{c,s}+b_{t,s}} \indicator{\widehat r_s=r}.
  \label{eq:checkpoint-coordinate-average}
\end{equation}
Each outcome updates one of the $2^d$ prefix bins. The same stream therefore recovers every coordinate at depth $d$. Changing the depth changes the suffix, readout target, and associated circuit template.

\subsection{Complex active-interface extension}
\label{subsec:checkpoint-state}

Before the suffix is applied, the reached interface consists of states whose lower $n-d-1$ qubits remain in $\ket0$. Its projector is
\begin{equation*}
  P_d := I_{q_n\cdots q_{n-d}} \otimes (\ket0\!\bra0)^{\otimes(n-d-1)}.
\end{equation*}
Both $\ket{\alpha_d}$ and $Q_{d,r}\ket{\alpha_d}$ lie in this subspace. A complex-compatible suffix therefore needs to agree with the separated phase-dressed suffix only on these reached columns:
\begin{equation*}
  B_{d,\mathbb C}P_d = \Dph B_dP_d.
\end{equation*}
This active-interface identity fixes exactly the action used by the checkpoint record. The unused unitary completion remains free. The separated choice $B_{d,\mathbb C}=\Dph B_d$ satisfies the same relation on the full space. Define the transported complex response by
\begin{equation*}
  \ket{\beta_d^{\mathbb C}} :=B_{d,\mathbb C}^\dagger O\ket{\psi^{\mathbb C}}.
\end{equation*}

\medskip
\noindent\textbf{One complex depth-$d$ template.}
An integrated realization applies $\HopfComplex$, $\ctrl(O)$, and
$B_{d,\mathbb C}^\dagger$, followed by the real checkpoint's $Y/Z/Y$
measurement and signed one-hot record.  The separated realization is
\begin{equation}
  \Uchk \longrightarrow \Dph \longrightarrow \ctrl(O) \longrightarrow \Dph^\dagger \longrightarrow B_d^\dagger.
  \label{eq:complex-checkpoint-circuit-order}
\end{equation}
The interface identity preserves the measured correlators and the average in Eq.~\eqref{eq:checkpoint-coordinate-average}; integrated and separated completions may act differently on the unused columns. At the final depth, $B_{n-1}=I$ and $P_{n-1}=I$. The real reverse block is the identity, and the complex reverse block is $\Dph^\dagger$. These templates supply the complex magnitude block. The direct phase template of Sec.~\ref{sec:complex-global} supplies the complementary phase-coordinate block.

Protocol~\ref{prot:checkpoint-gradient} gives the layer-local estimator one execution at a time. Different depths use separate streams because each has its own reverse suffix and measured target.

\begin{protocolcard}{Checkpoint method: \\one execution serves one depth}
  \label{prot:checkpoint-gradient}
  \protocolstep{CHOOSE A DEPTH}
  Fix $d$.  The checkpointed adjoint reverses the reflected branch only through
  the suffix below that layer.

  \protocolstep{ONE CHECKPOINT EXECUTION}
  Prepare the Hopf state, apply $\ctrl(O)$, and apply the chart-specific reverse
  block $B_d^\dagger$ or $B_{d,\mathbb C}^\dagger$.  Measure the ancilla and the
  depth-$d$ target in $Y$ and the active prefix in $Z$, obtaining
  $(b_c,b_t,\widehat r)$.

  \protocolstep{WHAT THAT EXECUTION MEANS}
  The prefix $\widehat r$ selects one coordinate at depth $d$, while the two
  $Y$ outcomes supply its sign.  The signed one-hot contribution
  $-2(-1)^{b_c+b_t}\mathbf e_{\widehat r}$ samples the complete depth-$d$
  gradient block.

  \protocolstep{AFTER MANY EXECUTIONS}
  Average the records to recover all $2^d$ derivatives at that depth.  Requested
  depths use separate streams because their reverse suffixes and targets differ.

  \protocolstep{COMPLEX COMPLETION}
  For a complete complex gradient, append the direct phase stream from the
  global-frame construction.
\end{protocolcard}

\subsection{Checkpoint guarantee and cost at a glance}

\medskip
\noindent\textbf{Guarantee.}
For either magnitude construction, with $\mu_d^{\mathrm{mag}}$ interpreted in
its corresponding chart,
\begin{equation}
  \E Z_d^{\mathrm{chk}}=\mu_d^{\mathrm{mag}},
  \label{eq:checkpoint-depth-gradient}
\end{equation}
and every record satisfies
\begin{equation}
  \|Z_d^{\mathrm{chk}}\|_2=2.
  \label{eq:checkpoint-fixed-norm}
\end{equation}
The checkpoint and global records therefore have the same statistical scale at a fixed depth. For one selected depth, the common count in Eq.~\eqref{eq:master-fixed-norm-shot-count} with failure probability $\eta=\delta$ suffices. Equation~\eqref{eq:complete-chart-shot-allocations} gives the complete-coordinate allocations. A complete real gradient uses one stream for each of the $n$ magnitude depths; a complete complex gradient adds the direct phase stream. Appendix~\ref{app:proof-real-checkpoint} proves the transported derivative, $Y$-quadrature sign, one-hot readout, interface substitution, and within-depth commutation. Appendix~\ref{app:four-qubit-checkpoints} gives the real and complex four-qubit depth-$2$ circuits. One selected depth is a block reverse pass; the complete all-depth schedule supplies Theorem~\ref{thm:hopf-quantum-backpropagation}(B).

\medskip
\noindent\textbf{Cost at a glance.}
A selected depth uses one template, $S_d$ repetitions, and a signed
$2^d$-bin histogram.  For $q$ selected depths, the checkpoint method sums the
repeated charges of $q$ templates.  A complete real schedule streams all
histograms in $\bigO(\sum_{d=0}^{n-1}S_d+N)$ classical operations and uses
$\bigO(N)$ auxiliary storage, producing $N-1$ entries because
$\sum_d2^d=N-1$.  At fixed accuracy and confidence, the complete real schedule
uses $\bigO(\log M\log\log M)$ executions with $M=N-1$; adding the direct
complex-phase histogram gives the same order with $M=2N-1$.
\section{Backpropagation accuracy, resources, and protocol choice}
\label{sec:resource-analysis}

\subsection{Reverse pullback and shared records}
\label{subsec:reverse-pullback}

The controlled reflection creates one objective response on state space. The inverse Hopf frame resolves that response along normalized magnitude directions, and the metric factors convert those frame components into coordinate derivatives. Writing $W_{\mathbb R}:=\WR$ and using the complex frame $W_{\mathbb C}$ from Sec.~\ref{sec:complex-global}, the component form of this pullback is, for $\mathbb K\in\{\mathbb R,\mathbb C\}$ and $1\leq j<N$,
\begin{equation}
  \partial_{\theta_j}\EO = 2\sqrt{g^{\mathbb K}_{j,j}} \Real\bra{b_{\lambda(j)}} W_{\mathbb K}^{\dagger}O \ket{\psi^{\mathbb K}}.
  \label{eq:hopf-coordinate-pullback}
\end{equation}
Controlled interference makes all addressed components in Eq.~\eqref{eq:hopf-coordinate-pullback} jointly readable. The checkpoint identity in Eq.~\eqref{eq:checkpoint-adjoint-identity}, together with the complex active-interface construction of Sec.~\ref{subsec:checkpoint-state}, implements the same coordinate pullback after transport to one selected depth. Complex phase directions are already localized in the computational basis, so their differential resolution is the identity. The global, checkpoint, and direct phase records therefore meet the operational record condition in Definition~\ref{def:state-coordinate-qbp}.

The global method reuses outcomes across all $n$ magnitude-depth records. Checkpointing instead uses a separate stream for each requested layer. In the complex chart, both constructions add a direct phase stream.

\subsection{Fixed-norm records and complete-chart accuracy}
\label{subsec:statistical-tools}

The width of a tree depth does not set its statistical scale. Every global magnitude-depth vector has norm two by Eq.~\eqref{eq:global-depth-fixed-norm}, and one global outcome supplies all of these depth records. Checkpoint and direct phase records are also norm-two signed one-hot vectors by Eqs.~\eqref{eq:checkpoint-fixed-norm} and \eqref{eq:complex-phase-onehot-vector}. The same concentration statement therefore controls every block, regardless of its number of coordinates.

\medskip
\noindent\textbf{Fixed-norm concentration.}
Let $Z_1,\ldots,Z_S\in\mathbb R^p$ be independent, identically distributed
vectors with $\|Z_t\|_2\leq B$ almost surely.  Set
$\mu:=\E Z_t$ and $\widehat\mu:=S^{-1}\sum_tZ_t$.  For every
$\eta\in(0,1)$,
\begin{equation}
  \Prob\!\left[ \|\widehat\mu-\mu\|_\infty > \frac{B\{1+\sqrt{2\log(1/\eta)}\}}{\sqrt S} \right] \leq\eta.
  \label{eq:fixed-norm-tail}
\end{equation}
For the norm-two Hopf records, define the common sufficient count
\begin{equation}
  S_\star(\eta) := \left\lceil \frac{4\{1+\sqrt{2\log(1/\eta)}\}^2}{\varepsilon^2} \right\rceil.
  \label{eq:master-fixed-norm-shot-count}
\end{equation}
Applying Eq.~\eqref{eq:fixed-norm-tail} to each complete depth vector and combining the depth-level failure events gives a simultaneous coordinate guarantee. This union bound does not require the depth records decoded from one outcome to be independent. Allocating the total failure probability across the required blocks gives
\begin{align}
  S_{\mathrm{mag}}^{\mathbb R} &:=S_\star(\delta/n), & S_{\mathrm{mag}}^{\mathbb C} &:=S_\star(\delta/(2n)),\nonumber\\
  S_{\mathrm{ph}}^{\mathbb C} &:=S_\star(\delta/2), & S_{\mathrm{glob}}^{\mathbb C} &:=S_{\mathrm{mag}}^{\mathbb C}+S_{\mathrm{ph}}^{\mathbb C},\nonumber\\
  S_{\mathrm{glob}}^{\mathbb R} &:=S_{\mathrm{mag}}^{\mathbb R}, & S_{\mathrm{chk}}^{\mathbb R} &:=nS_{\mathrm{mag}}^{\mathbb R},\nonumber\\
  && S_{\mathrm{chk}}^{\mathbb C} &:=nS_{\mathrm{mag}}^{\mathbb C}+S_{\mathrm{ph}}^{\mathbb C}.
  \label{eq:complete-chart-shot-allocations}
\end{align}
Consequently, for either chart $\mathbb K\in\{\mathbb R,\mathbb C\}$,
\begin{align}
  S_{\mathrm{glob}}^{\mathbb K} &= \bigO\!\left( \frac{1+\log(n/\delta)}{\varepsilon^2} \right),\nonumber\\
  S_{\mathrm{chk}}^{\mathbb K} &= \bigO\!\left( \frac{n\{1+\log(n/\delta)\}}{\varepsilon^2} \right).
  \label{eq:global-checkpoint-shot-scaling}
\end{align}
At fixed absolute coordinate accuracy and confidence, these become $\bigO(\log\log M)$ and $\bigO(\log M\log\log M)$, respectively. These are sufficient counts for the displayed allocations, not lower bounds on the number of executions. Per-execution work and the $M$-entry classical output are accounted for below. Appendix~\ref{app:proof-resource-criterion} proves the concentration bound, derives both orders explicitly, and records the corresponding complete-vector $\ell_2$ guarantee.

\subsection{Assigned Hopf CNOT work}
\label{subsec:assigned-ledger}

The resource comparison uses the direct-angle Hopf realization inherited from Ref.~\cite{lin2026hopf}. Each magnitude coordinate remains the physical rotation angle of its designated tree split, and each complex leaf phase remains a directly programmed phase angle. Keeping this coordinate-to-control map explicit makes the chart, inverse, metric, tangent frame, and counted physical controls part of one transparent geometry. The same cost functional is used for scalar and gradient programs. Assigned Hopf CNOT work is the assigned CNOT charge of Hopf-specific circuit blocks. For such a block $U$, let $C_{\mathrm H}(U)$ denote its assigned CNOT count. For a program $\mathcal P$, let $\mathsf W_{\mathrm H}^{\mathbb K}(\mathcal P;\varepsilon,\delta)$ denote the total charge of its Hopf forward and reverse blocks over all executions required for $(\varepsilon,\delta)$ accuracy.

The controlled reflection, its compiler workspace, final basis changes, measurement, and classical processing are reported as separate application and classical resources. Each matched scalar or gradient execution contains one call to the same $\ctrl(O)$. Adding the same nonnegative application-dependent controlled-reflection charge to every scalar and gradient execution preserves the overhead orders below up to a constant factor. The logarithmic matched-resource benchmark is therefore evaluated using assigned Hopf CNOT work and the matched storage comparison below. The explicit controlled two-level Hopf-gate charges assume all-to-all connectivity, zero clean work ancillas, exact one-qubit primitives, and the decompositions stated in Appendix~\ref{app:proof-compiler-ledger}.

For the designated schedules used in the theorem, set
\begin{align*}
  U_{\mathrm{fwd}}^{\mathbb R}&:=\Uchk, & U_{\mathrm{fwd}}^{\mathbb C}&:=\Dph\Uchk,\\
  c_{\mathrm{fwd}}^{\mathbb K} &:=C_{\mathrm H}(U_{\mathrm{fwd}}^{\mathbb K}).
\end{align*}
The matched scalar circuit uses this same forward preparation and the same $\ctrl(O)$ access. It omits the Hopf reverse block and measures the branch ancilla in $X$ to estimate $\EO$. Hoeffding's inequality~\cite{hoeffding1963probability} gives
\begin{equation}
  S_F(\varepsilon,\delta) = \left\lceil \frac{2}{\varepsilon^2}\log\frac{2}{\delta} \right\rceil,
  \label{eq:scalar-shot-count}
\end{equation}
and the corresponding assigned Hopf CNOT work is
\begin{equation}
  \mathsf W_{\mathrm H}^{\mathbb K}(\mathrm{scalar}) = S_F\,c_{\mathrm{fwd}}^{\mathbb K}.
  \label{eq:matched-scalar-hopf-work}
\end{equation}
The arguments $(\varepsilon,\delta)$ are suppressed in the remaining work formulas.

Let
\begin{align*}
  c_{\mathrm{fr}}^{\mathbb R} &:=C_{\mathrm H}(\WR^\dagger), & c_{\mathrm{fr}}^{\mathbb C} &:=C_{\mathrm H}(W_{\mathbb C}^\dagger),\\
  c_d^{\mathbb R} &:=C_{\mathrm H}(B_d^\dagger), & c_d^{\mathbb C} &:=C_{\mathrm H}(B_{d,\mathbb C}^\dagger).
\end{align*}
The repeated Hopf-block charges of the displayed schedules are
\begin{align}
  \mathsf W_{\mathrm H}^{\mathbb R}(\mathrm{global}) &= S_{\mathrm{mag}}^{\mathbb R} \left(c_{\mathrm{fwd}}^{\mathbb R}+c_{\mathrm{fr}}^{\mathbb R}\right), \nonumber\\
  \mathsf W_{\mathrm H}^{\mathbb C}(\mathrm{global}) &= S_{\mathrm{mag}}^{\mathbb C} \left(c_{\mathrm{fwd}}^{\mathbb C}+c_{\mathrm{fr}}^{\mathbb C}\right) + S_{\mathrm{ph}}^{\mathbb C}c_{\mathrm{fwd}}^{\mathbb C}, \nonumber\\
  \mathsf W_{\mathrm H}^{\mathbb R}(\mathrm{checkpoint}) &= \sum_{d=0}^{n-1} S_{\mathrm{mag}}^{\mathbb R} \left(c_{\mathrm{fwd}}^{\mathbb R}+c_d^{\mathbb R}\right), \nonumber\\
  \mathsf W_{\mathrm H}^{\mathbb C}(\mathrm{checkpoint}) &= \sum_{d=0}^{n-1} S_{\mathrm{mag}}^{\mathbb C} \left(c_{\mathrm{fwd}}^{\mathbb C}+c_d^{\mathbb C}\right) + S_{\mathrm{ph}}^{\mathbb C}c_{\mathrm{fwd}}^{\mathbb C}.
  \label{eq:hopf-work-block-formulas}
\end{align}
The complex magnitude and phase families therefore retain their separate execution counts and Hopf charges.

For the designated schedules,
\begin{equation}
  c_{\mathrm{fr}}^{\mathbb K} \leq2c_{\mathrm{fwd}}^{\mathbb K}, \qquad c_d^{\mathbb K} \leq c_{\mathrm{fwd}}^{\mathbb K} \quad(0\leq d<n).
  \label{eq:hopf-reverse-forward-comparability}
\end{equation}
The real frame bound follows because the final depth of $\Uchk$ already contains half as many gates with the same largest control count as $\WR$. Every inverse suffix is a subset of the depth layers of $\Uchk$. Adding the same diagonal phase layer to the designated complex forward and reverse blocks preserves the factors in Eq.~\eqref{eq:hopf-reverse-forward-comparability}. The exact sums appear in Appendix~\ref{app:proof-compiler-ledger}.

\medskip
\noindent\textbf{Matched storage.}
The designated scalar and gradient circuits use the same system register,
branch ancilla, forward-compiler workspace, and controlled-reflection
workspace.  Their blocks execute sequentially, so phase-layer and controlled-reflection
workspaces are reused.  Under the adopted Hopf-gate ledger, the addressed
controlled two-level gates use zero clean work ancillas.  A reverse schedule
adds streamed histograms and a complete gradient using $\bigO(M)$ classical
slots; the scalar input already contains the $M$ chart parameters.  Total
input/output-inclusive storage therefore grows by at most a constant factor
relative to the matched scalar program.

The record identities, norm-two concentration bound, and matched compiler comparison give the complete-gradient guarantee. One global magnitude stream serves all depths, whereas checkpoints use one stream per depth. A direct phase stream completes the complex chart.

\begin{theorem}[Quantum backpropagation and checkpointed reverse gradients
    with the Hopf ansatz]
  \label{thm:hopf-quantum-backpropagation}
  Let $\mathbb K\in\{\mathbb R,\mathbb C\}$, and let $M=N-1$ for $\mathbb K=\mathbb R$ and $M=2N-1$ for $\mathbb K=\mathbb C$. Under Assumption~\ref{ass:controlled-reflection}, the correctness and execution statements below hold for $n\geq1$. The assigned Hopf CNOT work and storage statements use the designated direct-angle separated schedules and assume $n\geq2$ to ensure a nonzero matched forward charge.

  \textbf{(A) Global-frame method.} The global magnitude pullback in Eq.~\eqref{eq:hopf-coordinate-pullback}, together with the direct phase family for $\mathbb K=\mathbb C$, implements the reverse state-coordinate task of Definition~\ref{def:state-coordinate-qbp}. With the allocations in Eq.~\eqref{eq:complete-chart-shot-allocations}, it produces an $(\varepsilon,\delta)$-accurate complete raw Hopf-coordinate gradient using
  \begin{equation*}
    S_{\mathrm{glob}}^{\mathbb K} = \bigO\!\left( \frac{1+\log(n/\delta)}{\varepsilon^2} \right)
  \end{equation*}
  independent executions. Under Eq.~\eqref{eq:hopf-reverse-forward-comparability}, the assigned Hopf CNOT work ratio relative to the matched scalar program satisfies
  \begin{equation}
    \frac{ \mathsf W_{\mathrm H}^{\mathbb K}(\mathrm{global}) }{ \mathsf W_{\mathrm H}^{\mathbb K}(\mathrm{scalar}) } = \bigO\!\left( \frac{1+\log(n/\delta)}{\log(2/\delta)} \right).
    \label{eq:hopf-global-work-ratio}
  \end{equation}

  \textbf{(B) Checkpoint method.} At each depth, the checkpointed adjoint implements the transported pullback in Eq.~\eqref{eq:checkpoint-adjoint-identity}, with the complex active-interface extension of Sec.~\ref{subsec:checkpoint-state}. A selected checkpoint therefore implements the reverse state-coordinate task of Definition~\ref{def:state-coordinate-qbp} for its depth block. With the allocations in Eq.~\eqref{eq:complete-chart-shot-allocations}, the complete all-depth schedule, together with the direct phase family for $\mathbb K=\mathbb C$, produces the same $(\varepsilon,\delta)$-accurate complete raw Hopf-coordinate gradient using
  \begin{equation*}
    S_{\mathrm{chk}}^{\mathbb K} = \bigO\!\left( \frac{n\{1+\log(n/\delta)\}}{\varepsilon^2} \right)
  \end{equation*}
  independent executions. Under Eq.~\eqref{eq:hopf-reverse-forward-comparability},
  \begin{equation}
    \frac{ \mathsf W_{\mathrm H}^{\mathbb K}(\mathrm{checkpoint}) }{ \mathsf W_{\mathrm H}^{\mathbb K}(\mathrm{scalar}) } = \bigO\!\left( \frac{n\{1+\log(n/\delta)\}}{\log(2/\delta)} \right).
    \label{eq:hopf-checkpoint-work-ratio}
  \end{equation}

  Under the designated direct-angle forward-charge condition $c_{\mathrm{fwd}}^{\mathbb K}=\bigO(M\log M)$, the same cost functional gives, at fixed accuracy and confidence,
  \begin{align}
    \mathsf W_{\mathrm H}^{\mathbb K}(\mathrm{scalar}) &=\bigO(M\log M),\nonumber\\
    \mathsf W_{\mathrm H}^{\mathbb K}(\mathrm{global}) &=\bigO(M\log M\log\log M),\nonumber\\
    \mathsf W_{\mathrm H}^{\mathbb K}(\mathrm{checkpoint}) &=\bigO(M\log^2 M\log\log M).
    \label{eq:hopf-absolute-work-scaling}
  \end{align}
  For the designated schedules, the matched-storage argument above gives $\bigO(1)$ total input/output-inclusive storage overhead. The additional classical histogram and output storage is $\bigO(M)$. Since $M=2^n-1$ in the real case and $M=2^{n+1}-1$ in the complex case, the preceding bounds have the stated dependence on $M$. Relative to the matched scalar program, the global-frame method has $\bigO(\log\log M)$ overhead in assigned Hopf CNOT work and therefore meets the logarithmic matched-resource benchmark. The displayed complete-checkpoint overhead is $\bigO(\log M\log\log M)$. Within the same resource specialization it remains polylogarithmic in the broader sense of Ref.~\cite{abbas2023quantum}, but the present bound does not meet the stricter logarithmic benchmark.
\end{theorem}

\noindent\textbf{Coordinate and geometric outputs.}
The theorem controls raw Hopf-coordinate derivatives in the norm of
Definition~\ref{def:complete-gradient}.  For an active real magnitude
coordinate, define the normalized-frame and inverse-metric outputs by
\begin{equation*}
  \chi_j:=\frac{\partial_{\theta_j}\EO}{\sqrt{g^{\mathbb R}_{j,j}}}, \qquad \nu_j:=\frac{\partial_{\theta_j}\EO}{g^{\mathbb R}_{j,j}}.
\end{equation*}
At a parameter point with $0<g^{\mathbb R}_{k,k}<(\varepsilon/2)^2$, choose a Hermitian unitary that swaps $\ket{\psi^{\mathbb R}}$ with $\ket{e_k^{\mathbb R}}$ and hold it fixed during differentiation. Then $\partial_{\theta_k}E_O=2\sqrt{g^{\mathbb R}_{k,k}}<\varepsilon$ but $\chi_k=2$. Thus raw coordinate and normalized-frame accuracy are distinct near small metric weights; Appendix~\ref{app:proof-resource-criterion} gives the witness and conditioning bounds.

The theorem uses the direct-angle ledger, in which each Hopf coordinate remains one directly programmed physical angle. Appendix~\ref{app:compiler-robustness} separately gives an exact multiplexed realization that recombines those angles while preserving the estimator and $\bigO(N)$ matched forward/frame scaling. This alternative has its own compiler ledger; it does not replace the theorem's accounting. Finite and proof details appear in Appendices~\ref{app:four-qubit-ledger}, \ref{app:proof-compiler-ledger}, and \ref{app:proof-hopf-qbp}.

\medskip
\noindent\textbf{Classical processing.}
These operation counts use a word-RAM whose words can store one $n$-bit
outcome label, so reading or incrementing one histogram address costs
$\bigO(1)$.  Global decoding streams an $N$-bin signed histogram and applies
one Walsh--Hadamard transform, requiring
$\bigO(S_{\mathrm{glob}}^{\mathbb K}+N\log N)$ operations.  Checkpoint and
direct phase records use signed histograms, requiring
$\bigO(S_{\mathrm{chk}}^{\mathbb K}+N)$ operations for a complete schedule.
A complete gradient contains $M$ entries, so linear output work is intrinsic
to the requested output.

\subsection{Choosing the reverse protocol}
\label{sec:comparison}

At a fixed depth, $Z_d^{\mathrm{fr}}$ and $Z_d^{\mathrm{chk}}$ are unbiased for $\mu_d^{\mathrm{mag}}$ and both have norm two. They differ in how many depths one template and record stream serves.

The global frame is the natural choice when many or all depths are requested and its inverse is available. Every all-$X$ outcome contributes to every depth; one signed histogram and Walsh transform decode the magnitude block. The complex chart adds the direct phase stream.

Checkpoints are useful when one or a few depths are active and their suffixes or interfaces are easier to compile. Each depth has its own reverse block, readout target, and $2^d$-bin histogram, whose average recovers the complete depth block. The final real depth needs no reverse suffix.

For $q$ requested depths $d_1,\ldots,d_q$, checkpoints use $q$ templates. They are favorable when their repeated compiled records cost less than the full-frame records or when a layer-local update is the intended output. Appendix~\ref{app:four-qubit-ledger} gives the finite four-qubit comparison. Both methods resolve the same normalized differential: globally in the addressed frame or locally at an active interface.

\section{Conclusion}
\label{sec:conclusion}

The Hopf tree carries an objective response backward as well as a state forward. Under phase-calibrated controlled-reflection access, the global construction maps the reflected response to computational addresses for all normalized magnitude directions. One shared record stream recovers the complete magnitude block, and a separate signed leaf record supplies the complex phase derivatives. The magnitude stream alone completes the real chart; the two streams together complete the complex chart. Checkpointing provides the layer-local alternative, reversing only the suffix below a chosen depth to recover that layer's derivative block.

The reverse construction uses the same direct-angle geometry as the Hopf representation. Each magnitude coordinate defines one nested amplitude split, is programmed as the physical angle of that split, and determines one differential direction. Each complex leaf phase is likewise a directly programmed control. The same recursion prepares the state; its classical inverse recovers coordinates from known amplitudes, and its diagonal metric supplies the corresponding weights. Representation, circuit control, and differentiation therefore remain in one language. The circuit structure that defines the model also supplies the directions in which its objective response is decoded.

At fixed additive coordinate accuracy and confidence, sufficient execution counts are $\bigO(\log\log M)$ for the global method and $\bigO(\log M\log\log M)$ for the complete checkpoint schedule. The global gain is record reuse across magnitude coordinates, not a smaller classical output. Under the designated direct-angle ledger, the global construction meets the logarithmic matched-resource benchmark relative to a scalar evaluation using the same forward preparation and controlled reflection. Checkpointing is complementary for a selected layer with a less costly reverse suffix. These guarantees concern raw coordinate derivatives: the $M$-entry output, classical decoding, and conditioning of relative or metric-rescaled accuracy remain separate costs and requirements.

Applications require the stated coherent objective interface or a suitable reflection-sum decomposition. Within that condition, complete coordinate updates may support variational scientific computing, pure-state learning, and probe or decision-state design, while checkpoint records may support layer-local quantum-control updates. These are uses of the gradient interface, not application-level performance guarantees. The data encoding, controlled-reflection compiler, aggregation rule, and optimizer must still be specified for each problem.

The principal open question is whether the coherent objective interface can be weakened without losing the shared reverse record. Alternative interference or measurement constructions may expose the same transition moments with different control requirements. A second direction is to adapt the Hopf organization to structured state families whose relevant degrees of freedom occupy only part of the full pure-state sphere. The two papers thus make the direct-angle Hopf geometry bidirectional: the compass maps coordinates to any pure state, and quantum backpropagation through the controlled-reflection mirror returns the objective response to them.

\begin{acknowledgments}
The authors are sincerely grateful to Allan Tosta, Andr'e J. Ferreira-Martins,
Thiago O. Maciel, and Tobias Haug for thoughtful discussions and perspectives
that helped shape this work, and to Leandro Aolita and the Technology Innovation
Institute (TII) for their valued support of the project. Special thanks are due to Giancarlo Camilo and Renato M. S.
Farias for sustained guidance, encouragement, and perspective.

The authors also thank Xin Wang for his valuable assistance with the final revision of the manuscript. The work was partially supported by the National Natural Science Foundation of China
(Grant No.~12447107), the Guangdong Provincial Quantum Science
Strategic Initiative (Grant No.~GDZX2403008, GDZX2503001), and the CCF-Tencent
Rhino-Bird Open Research Fund.
\end{acknowledgments}

\makeatletter
\@booleantrue\raggedcolumn@sw
\makeatother
\section*{Data and code availability}

No external datasets were used in this work. The Hopf-QBP
repository~\cite{lin2026hopfqbpcode} provides independent NumPy reference
implementations and Qibo~\cite{efthymiou2022qibo} circuit builders and decoders,
together with deterministic tests supported by continuous integration.
The repository also contains the direct-angle assigned Hopf CNOT ledger,
exact-recompilation and statistical analyses, and readout checks.
The accompanying reproducibility instructions describe how to run the
implementations and associated checks. Circuit checks were performed
using exact Qibo output distributions.

\vspace{1cm}

\section*{Generative AI disclosure}

This paper began in a midnight conversation between the first author and OpenAI's ChatGPT, using GPT-5.5 Pro and GPT-5.6 Pro, and developed through many rounds of conjecture, calculation, criticism, and revision. ChatGPT assisted in exploring constructions, developing arguments, checking implementations, and shaping the manuscript. GPT-6 was also used during the final revision phase. The dialogue served as a working blackboard: ideas were tested, discarded, or rebuilt as the mathematics required. The first author chose the questions, set the scope, selected the paths, and verified the final results. All authors reviewed the arguments, calculations, implementations, and manuscript and accept responsibility.

\FloatBarrier

\bibliography{references}

\onecolumngrid
\clearpage
\makeatletter
\@booleanfalse\raggedcolumn@sw
\makeatother

\appendix
\FloatBarrier

\section{Four-qubit implementation companion}
\label{app:four-qubit-example}

This appendix gives the $n=4$, $N=16$ instance. The figures display the native preparations, depth completion, addressed frames, global records, and depth-$2$ checkpoints. Three tables supplement the figures and general formulas.

\subsection{How to read the four-qubit construction}
\label{app:four-qubit-guide}

The wires are drawn with $q_1$ at the top and $q_4$ at the bottom, while computational-basis labels use
\begin{equation*}
  \ket{b_\ell}=\ket{q_4q_3q_2q_1}, \qquad \ell=(q_4q_3q_2q_1)_2.
\end{equation*}
A filled control fixes a basis bit to one, and an open control fixes it to zero. Rotations follow Eq.~\eqref{eq:paper-Ry-convention}; a library using $\operatorname{RY}(\varphi)=e^{-\iu\varphi Y/2}$ therefore receives software angle $2\theta$. The dashed blocks use the full-unitary, state-column, and active-interface contracts of Eq.~\eqref{eq:three-circuit-contracts}.

The coordinate blocks are
\begin{align*}
  \boldsymbol\theta_{\mathrm{mag}} &=(\theta_1,\ldots,\theta_{15}), & \boldsymbol\theta_{\mathrm{ph}} &=(\theta_{16},\ldots,\theta_{31}),\\
  \phi_\ell&:=\theta_{16+\ell}, &0&\leq\ell<16.
\end{align*}

\subsection{Native real and complex preparation}
\label{app:four-qubit-native}

The breadth-first internal indices are $1,\ldots,15$, and the leaves are $16+\ell$. In Fig.~\ref{fig:app-hopf4-tree}, the uncircled labels give the internal-node indices and the circled labels give the native gate order of Ref.~\cite{lin2026hopf}.

\begin{center}
  \begin{minipage}{\textwidth}
    \centering
    \resizebox{15cm}{!}{
      \begin{tikzpicture}[
          level distance=3.2cm,
          level 1/.style = {sibling distance=5cm},
          level 2/.style = {sibling distance=8cm},
          level 3/.style = {sibling distance=2cm},
          level 4/.style = {sibling distance=1cm},
          every node/.style = {draw=none, fill=OliveGreen!20, circle, minimum size=1cm, font=\bfseries},
          edge from parent/.style = {draw, ->, thick},
          label1/.style = {below=1.2cm, font=\normalsize, draw=none,fill=none},
          label2/.style = {below=1.4cm,circle,inner sep=0.6pt,minimum size=0.6pt, font=\normalsize, draw,fill=none},
          label3/.style = {below=1.4cm,circle,inner sep=0.6pt,minimum size=0.6pt, font=\scriptsize, draw,fill=none},
          label4/.style = {below=1.4cm,circle,inner sep=0.6pt,minimum size=0.6pt, line width=0.8pt,font=\normalsize, draw,fill=none},
          label5/.style = {below=0.6cm, font=\normalsize, draw=none,fill=none, color=NavyBlue}
        ]

        \node (root) {0000}
        child {node (node0) {0000}
            child  [shift={(0.375cm,0)}] {node (node00) {0000}
                child [shift={(-1cm,0)}] {node (node000) {0000}
                    child [shift={(-0.75cm,0)}] {node (node0000) {0000} }
                    child [shift={(-0.5cm,0)}] {node (node0001) {0001} }
                      [shift={(-0.375cm,0)}]  node[label1] {$\theta_{8}$}
                      node[label2, below=0.6cm] {4}
                  }
                child [shift={(-1cm,0)}] {node (node001) {0010}
                    child [shift={(-0.25cm,0)}] {node (node0010) {0010} }
                    child {node (node0011) {0011} }
                      node[label1] {$\theta_{9}$}
                      node[label2, below=0.6cm] {5}
                  }
                  [shift={(-0.5cm,0)}] node[label1] {$\theta_{4}$}
                  node[label2, below=0.6cm] {3}
              }
            child  [shift={(-4cm,0)}] {node (node01) {0100}
                child [shift={(0.375cm,0)}] {node (node010) {0100}
                    child [shift={(-0.75cm,0)}] {node (node0100) {0100} }
                    child [shift={(-0.5cm,0)}] {node (node0101) {0101} }
                      [shift={(-0.375cm,0)}] node[label1] {$\theta_{10}$}
                      node[label2, below=0.6cm] {7}
                  }
                child [shift={(-0.375cm,0)}] {node (node011) {0110}
                    child [shift={(0.5cm,0)}] {node (node0110) {0110}}
                    child [shift={(0.75cm,0)}] {node (node0111) {0111}}
                      [shift={(0.375cm,0)}] node[label1] {$\theta_{11}$}
                      node[label3, below=0.6cm] {11}
                  }
                  node[label1] {$\theta_{5}$}
                  node[label2, below=0.6cm] {6}
              }
              [shift={(-1cm,0)}] node[label1] {$\theta_{2}$}
              node[label2, below=0.6cm] {2}
          }
        child  {node (node1) {1000}
            child [shift={(4cm,0)}]  {node (node10) {1000}
                child [shift={(0.375cm,0)}]{node (node100) {1000}
                    child [shift={(-0.75cm,0)}] {node (node1000) {1000}}
                    child [shift={(-0.5cm,0)}]{node (node1001) {1001}}
                      [shift={(-0.375cm,0)}] node[label1] {$\theta_{12}$}
                      node[label3, below=0.6cm] {10}
                  }
                child [shift={(-0.375cm,0)}] {node (node101) {1010}
                    child [shift={(0.5cm,0)}]{node (node1010) {1010}}
                    child [shift={(0.75cm,0)}] {node (node1011) {1011}}
                      [shift={(0.375cm,0)}] node[label1] {$\theta_{13}$}
                      node[label3, below=0.6cm] {12}
                  }
                  node[label1] {$\theta_{6}$}
                  node[label2, below=0.6cm] {9}
              }
            child [shift={(-0.375cm,0)}] {node (node11) {1100}
                child [shift={(1cm,0)}] {node (node110) {1100}
                    child {node (node1100) {1100}}
                    child [shift={(0.25cm,0)}] {node (node1101) {1101}}
                      [shift={(0.125cm,0)}] node[label1] {$\theta_{14}$}
                      node[label3, below=0.6cm] {14}
                  }
                child [shift={(1cm,0)}]{node (node111) {1110}
                    child [shift={(0.5cm,0)}] {node (node1110) {1110}}
                    child [shift={(0.75cm,0)}] {node (node1111) {1111}}
                      [shift={(0.375cm,0)}] node[label1] {$\theta_{15}$}
                      node[label3, below=0.6cm] {15}
                  }
                  [shift={(0.5cm,0)}] node[label1] {$\theta_{7}$}
                  node[label3, below=0.6cm] {13}
              }
              [shift={(1cm,0)}] node[label1] {$\theta_{3}$}
              node[label2, below=0.6cm] {8}
          }  node[label1] {$\theta_{1}$}
          node[label4, below=0.6cm] {1};

        \path (node0000) node[label5] {$\theta_{16}$};
        \path (node0001) node[label5] {$\theta_{17}$};
        \path (node0010) node[label5] {$\theta_{18}$};
        \path (node0011) node[label5] {$\theta_{19}$};
        \path (node0100) node[label5] {$\theta_{20}$};
        \path (node0101) node[label5] {$\theta_{21}$};
        \path (node0110) node[label5] {$\theta_{22}$};
        \path (node0111) node[label5] {$\theta_{23}$};
        \path (node1000) node[label5] {$\theta_{24}$};
        \path (node1001) node[label5] {$\theta_{25}$};
        \path (node1010) node[label5] {$\theta_{26}$};
        \path (node1011) node[label5] {$\theta_{27}$};
        \path (node1100) node[label5] {$\theta_{28}$};
        \path (node1101) node[label5] {$\theta_{29}$};
        \path (node1110) node[label5] {$\theta_{30}$};
        \path (node1111) node[label5] {$\theta_{31}$};

      \end{tikzpicture}
    }
    
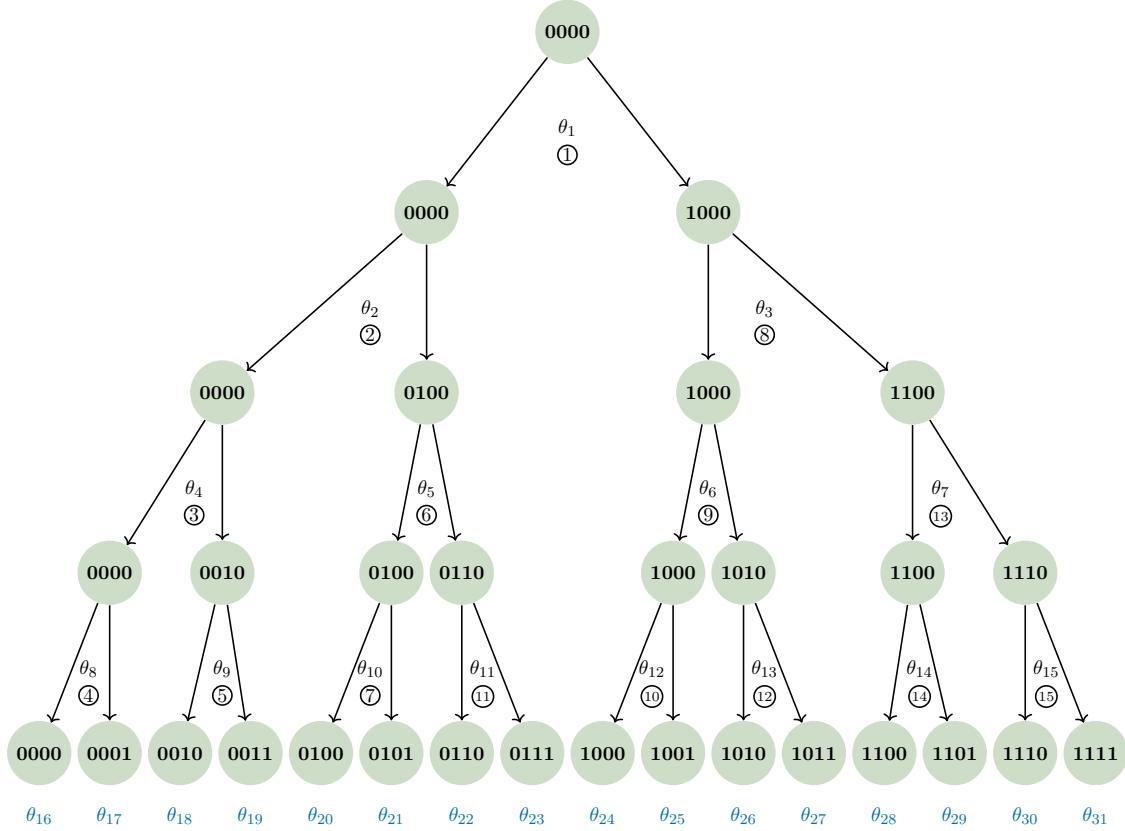
\captionof{figure}{Four-qubit Hopf tree in the first-paper convention.  A left branch contributes $\cos\theta_j$ and a right branch contributes $\sin\theta_j$.  The internal labels are the magnitude coordinates, the blue labels beneath the leaves are the complex phase coordinates, and the circled numbers give the native gate-application order.  Padded binary labels identify subtree anchors and leaves.}
    \label{fig:app-hopf4-tree}
  \end{minipage}
\end{center}

The promoted two-leaf completion used by the complex circuit is
\begin{equation*}
  R_{\mathbb C}(\theta_a,\theta_b,\theta_c) =
  \begin{pmatrix}
    e^{\iu\theta_b}\cos\theta_a & -e^{-\iu\theta_c}\sin\theta_a \\
    e^{\iu\theta_c}\sin\theta_a & \phantom{-}e^{-\iu\theta_b}\cos\theta_a
  \end{pmatrix}.
\end{equation*}
Its first column assigns phases $\theta_b$ and $\theta_c$ to two sibling leaves. The native circuit follows the circled tree order; its complete control schedule is inherited from Ref.~\cite{lin2026hopf} and implemented in the accompanying validation code.

\begin{center}
  \begin{minipage}{\textwidth}
    \centering
    \maybegraphic[scale=0.25]{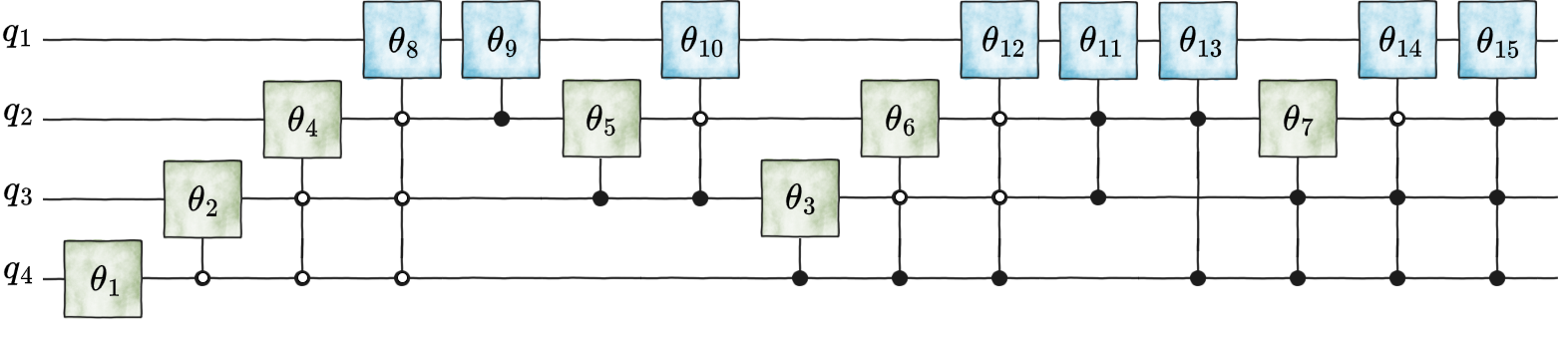}
    \captionof{figure}{Native four-qubit Hopf ansatz circuit.  For $\HopfReal(\boldsymbol\theta_{\mathrm{mag}})$, every box is $R_y(\theta_j)$.  For $\HopfComplex(\boldsymbol\theta)$, each final-layer gate $j=8+r$, $0\leq r<8$, is promoted to $R_{\mathbb C}(\theta_{8+r},\theta_{16+2r},\theta_{17+2r})$.  The blue boxes identify that promoted final layer.}
    \label{fig:app-native-hopf-circuit}
  \end{minipage}
\end{center}

The initialized native columns prepare $\ket{\psi^{\mathbb R}}$ and $\ket{\psi^{\mathbb C}}$; the complete state-column contracts are collected below with the depth and frame completions.

\subsection{Metric factors and tangent-state registry}
\label{app:four-qubit-metric}

The single registry below combines the incoming metric factor, the real-frame anchor and marker, and the checkpoint prefix and target for every internal node.

\begin{center}
  \begin{minipage}{0.82\textwidth}
    \centering
    \scriptsize
    \setlength{\tabcolsep}{6pt}
    \begin{tabular}{c c c c c}
      \toprule
      $j$ & $(d,r)$ & $\sqrt{g^{\mathbb R}_{j,j}}$ & $\ell_0(j)$ & $\lambda(j)$\\
      \midrule
      1  & $(0,0)$ & $1$ & 0  & 8\\
      2  & $(1,0)$ & $\cos\theta_1$ & 0  & 4\\
      3  & $(1,1)$ & $\sin\theta_1$ & 8  & 12\\
      4  & $(2,0)$ & $\cos\theta_1\cos\theta_2$ & 0  & 2\\
      5  & $(2,1)$ & $\cos\theta_1\sin\theta_2$ & 4  & 6\\
      6  & $(2,2)$ & $\sin\theta_1\cos\theta_3$ & 8  & 10\\
      7  & $(2,3)$ & $\sin\theta_1\sin\theta_3$ & 12 & 14\\
      8  & $(3,0)$ & $\cos\theta_1\cos\theta_2\cos\theta_4$ & 0  & 1\\
      9  & $(3,1)$ & $\cos\theta_1\cos\theta_2\sin\theta_4$ & 2  & 3\\
      10 & $(3,2)$ & $\cos\theta_1\sin\theta_2\cos\theta_5$ & 4  & 5\\
      11 & $(3,3)$ & $\cos\theta_1\sin\theta_2\sin\theta_5$ & 6  & 7\\
      12 & $(3,4)$ & $\sin\theta_1\cos\theta_3\cos\theta_6$ & 8  & 9\\
      13 & $(3,5)$ & $\sin\theta_1\cos\theta_3\sin\theta_6$ & 10 & 11\\
      14 & $(3,6)$ & $\sin\theta_1\sin\theta_3\cos\theta_7$ & 12 & 13\\
      15 & $(3,7)$ & $\sin\theta_1\sin\theta_3\sin\theta_7$ & 14 & 15\\
      \bottomrule
    \end{tabular}
    \captionof{table}{Four-qubit internal-node registry.  For $j=2^d+r$, the checkpoint prefix is $r$, its target is $q_{4-d}$, $\ell_0(j)$ is the left-child anchor, and $\lambda(j)$ addresses the normalized local complement.}
    \label{tab:app-four-qubit-node-registry}
  \end{minipage}
\end{center}

The same incoming factors govern the complex magnitude block by Eq.~\eqref{eq:complex-real-magnitude-metric}. For the nontrivial node $5$, the normalized local complement is
\begin{equation*}
  \ket{e_5^{\mathbb R}} =-\sin\theta_5 \left(\cos\theta_{10}\ket{0100}+\sin\theta_{10}\ket{0101}\right) +\cos\theta_5 \left(\cos\theta_{11}\ket{0110}+\sin\theta_{11}\ket{0111}\right),
\end{equation*}
and the corresponding incoming factor is
\begin{equation}
  \sqrt{g^{\mathbb R}_{5,5}} =\cos\theta_1\sin\theta_2.
  \label{eq:app-node5-path-factor}
\end{equation}

\subsection{Depth completion $U_{\mathrm{chk}}$}
\label{app:four-qubit-depth-completion}

The four depth layers are
\begin{align*}
  U_0&=R_y^{(q_4)}(\theta_1),\\
  U_1&=\sum_{r=0}^1\ket r\!\bra r_{q_4} \otimes R_y^{(q_3)}(\theta_{2+r}),\\
  U_2&=\sum_{r=0}^3\ket r\!\bra r_{q_4q_3} \otimes R_y^{(q_2)}(\theta_{4+r}),\\
  U_3&=\sum_{r=0}^7\ket r\!\bra r_{q_4q_3q_2} \otimes R_y^{(q_1)}(\theta_{8+r}).
\end{align*}
Thus
\begin{equation*}
  \Uchk=U_3U_2U_1U_0, \qquad A_2:=U_2U_1U_0, \qquad B_2:=U_3, \qquad \Uchk=B_2A_2.
\end{equation*}

\begin{center}
  \begin{minipage}{\textwidth}
    \centering
    \maybegraphic[scale=0.25]{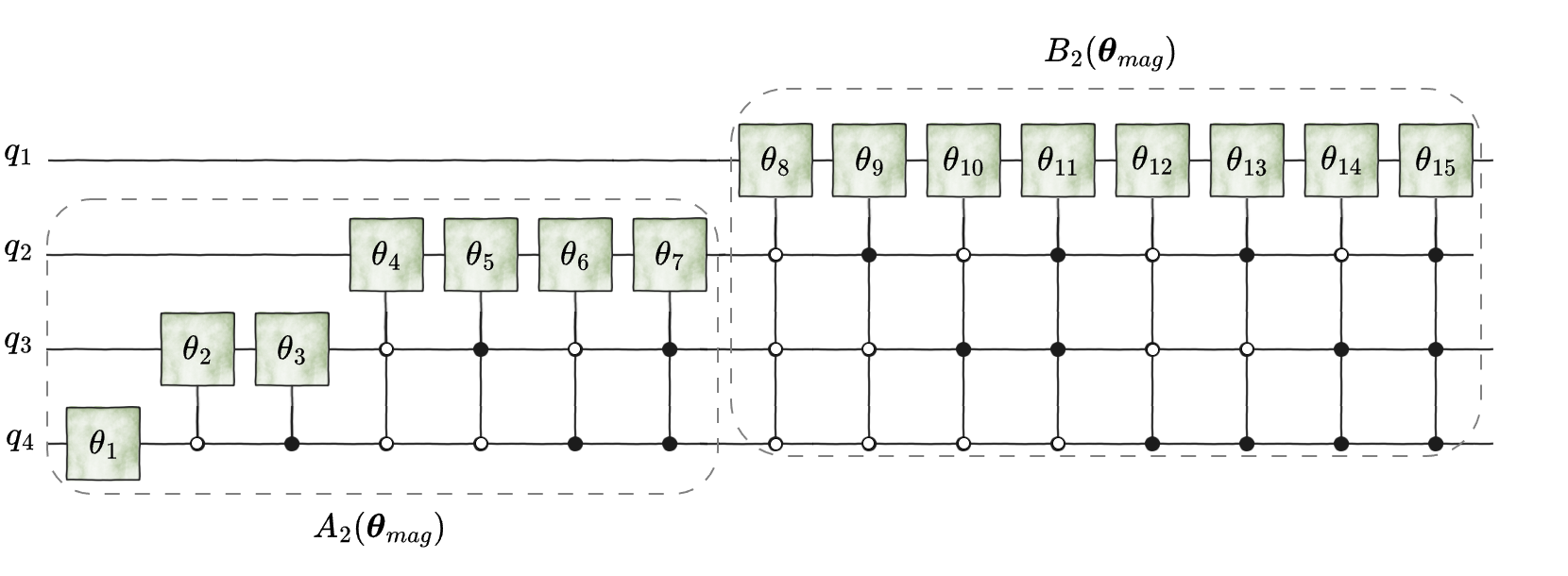}
    \captionof{figure}{Four-qubit depth completion.  The left bracket is
      $A_2(\boldsymbol\theta_{\mathrm{mag}})=U_2U_1U_0$ and the right bracket is
      $B_2(\boldsymbol\theta_{\mathrm{mag}})=U_3$, giving the full-unitary factorization
      $\Uchk=B_2A_2$.  Open controls select zero-valued prefix bits.
      Section~\ref{app:four-qubit-checkpoints} constructs the corresponding complex
      depth-$2$ completion.}
    \label{fig:app-Uchk-gatelevel}
  \end{minipage}
\end{center}

The native and depth completions share the initialized real state column in Eq.~\eqref{eq:app-three-real-state-columns}, although their actions on the remaining columns generally differ.

\subsection{Real addressed frame $W_{\mathbb R}$}
\label{app:four-qubit-real-frame}

For $j=2^d+r$, Table~\ref{tab:app-four-qubit-node-registry} gives the active pair $\ket{b_{\ell_0(j)}}\leftrightarrow \ket{b_{\lambda(j)}}$. The target is $q_{4-d}$, and the controls fix the remaining bits to the anchor pattern. Applying parent layers before descendants gives
\begin{equation*}
  \WR(\boldsymbol\theta_{\mathrm{mag}})=V_3V_2V_1V_0,
\end{equation*}
with $V_0$ applied first in state-update order.

\begin{center}
  \begin{minipage}{\textwidth}
    \centering
    \maybegraphic[scale=0.25]{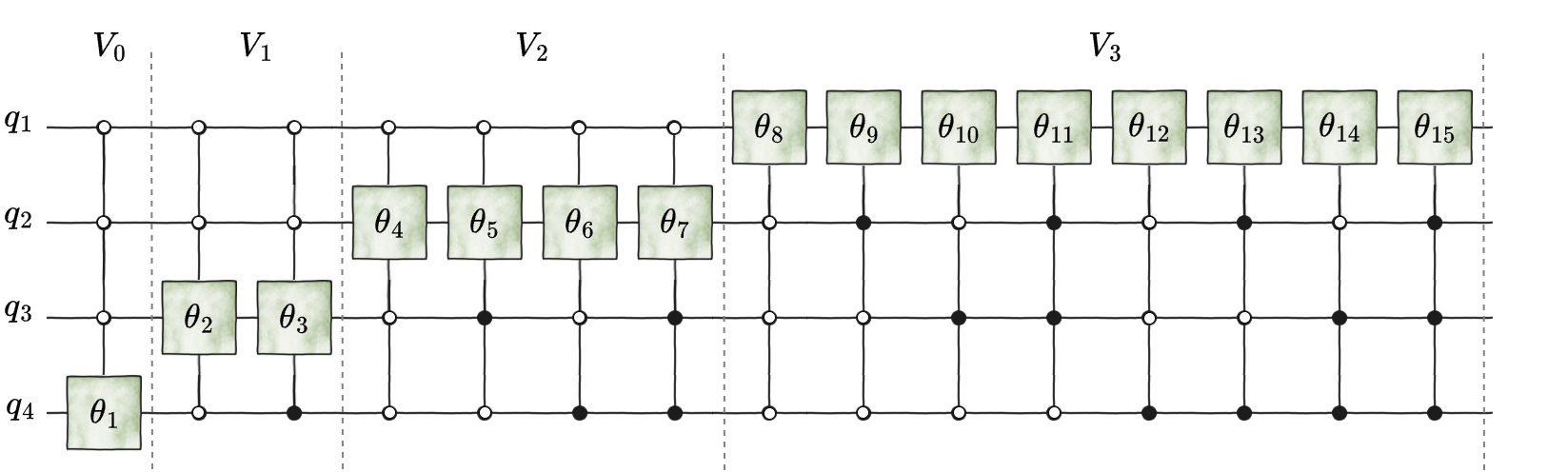}
    \captionof{figure}{Fully addressed four-qubit real-magnitude frame $\WR$, with $V_0$ acting first.
      Every non-target qubit is fixed by a closed or open control, so each gate
      preserves the required state-and-complement column action.  A box labeled
      $\theta_j$ denotes Eq.~\eqref{eq:paper-Ry-convention}.  The same addressed
      gate skeleton is used for the complex-magnitude frame constructed in
      Sec.~\ref{app:four-qubit-complex-frame}.  Circuit layouts use the visual
      language of the quantum-circuit-drawio library~\cite{wilkens2023quantumcircuitlibrary}.}
    \label{fig:app-frame-gatelevel}
  \end{minipage}
\end{center}

For node $5$, the addressed gate first gives
\[
  G_5\ket{0110} =-\sin\theta_5\ket{0100}+\cos\theta_5\ket{0110}.
\]
The descendant gates then map $\ket{0110}$ to the previously displayed $\ket{e_5^{\mathbb R}}$. Together with Eq.~\eqref{eq:app-node5-path-factor}, this gives $\partial_{\theta_5}\ket{\psi^{\mathbb R}} =\cos\theta_1\sin\theta_2\ket{e_5^{\mathbb R}}$. The complete frame obeys
\begin{equation*}
  \WR\ket{0000}=\ket{\psi^{\mathbb R}}, \qquad \WR\ket{b_{\lambda(j)}}=\ket{e_j^{\mathbb R}}, \quad 1\leq j\leq15.
\end{equation*}
Thus
\begin{equation}
  \HopfReal\ket{0000} = \Uchk\ket{0000} = \WR\ket{0000} = \ket{\psi^{\mathbb R}},
  \label{eq:app-three-real-state-columns}
\end{equation}
while the three circuits generally give different full-unitary completions.

\subsection{Complex magnitude frame $W_{\mathbb C}$}
\label{app:four-qubit-complex-frame}

Define
\begin{align*}
  \Dph(\boldsymbol\theta_{\mathrm{ph}}) &:= \sum_{\ell=0}^{15}e^{\iu\phi_\ell} \ket{b_\ell}\!\bra{b_\ell},\\
  W_{\mathbb C}(\boldsymbol\theta) &:= \Dph(\boldsymbol\theta_{\mathrm{ph}}) \WR(\boldsymbol\theta_{\mathrm{mag}}).
\end{align*}
Then
\begin{align}
  \HopfComplex(\boldsymbol\theta)\ket{0000} &= \Dph\Uchk\ket{0000} = W_{\mathbb C}(\boldsymbol\theta)\ket{0000} = \ket{\psi^{\mathbb C}},\nonumber\\
  W_{\mathbb C}\ket{b_{\lambda(j)}} &=\ket{e_j^{\mathbb C}}, \qquad 1\leq j\leq15,\nonumber\\
  W_{\mathbb C}^\dagger &=\WR^\dagger\Dph^\dagger.
  \label{eq:app-WC-frame-columns}
\end{align}
The first line is a state-column contract for the native and depth forwards; the remaining two identities fix the complete complex-magnitude frame used by the global record.

\subsubsection{Addressed $R_{\mathbb C}$ compilation up to a compiler global phase}
\label{app:four-qubit-RC-compiler}

Let
\begin{equation*}
  \bar\phi:=\frac1{16}\sum_{\ell=0}^{15}\phi_\ell, \qquad \widehat\phi_\ell:=\phi_\ell-\bar\phi,
\end{equation*}
and define centered subtree sums by
\begin{equation*}
  \Phi_{16+\ell}:=\widehat\phi_\ell, \qquad \Phi_j:=\Phi_{2j}+\Phi_{2j+1}.
\end{equation*}
On the same active pair and with the same controls as the real gate, use
\begin{equation}
  \widetilde G_j := R_{\mathbb C}\!\left(\theta_j,-\Phi_{2j+1},-\Phi_{2j}\right).
  \label{eq:app-WC-RC-gate}
\end{equation}
With the same parent-before-descendant order,
\begin{equation}
  \boxed{ \widetilde W_{\mathbb C}^{(R_{\mathbb C})} =e^{-\iu\bar\phi}W_{\mathbb C}. }
  \label{eq:app-WC-RC-compiler-identity}
\end{equation}
The compiler global phase multiplies every column equally and therefore leaves the protocol probabilities unchanged.

The recursion determines every gate directly. For example, the second and third arguments of $\widetilde G_1$ are
\[
  \left(-\sum_{\ell=8}^{15}\widehat\phi_\ell, -\sum_{\ell=0}^{7}\widehat\phi_\ell\right),
\]
while those of $\widetilde G_5$ are
\[
  \left(-\widehat\phi_6-\widehat\phi_7, -\widehat\phi_4-\widehat\phi_5\right).
\]
All remaining arguments follow from the same two subtree sums.

\subsection{Global gradient records}
\label{app:four-qubit-global-records}

The real circuit uses native preparation for the initialized state column and $\WR^\dagger$ for the complete reverse frame.

\begin{center}
  \begin{minipage}{\textwidth}
    \centering
    \maybegraphic[scale=0.25]{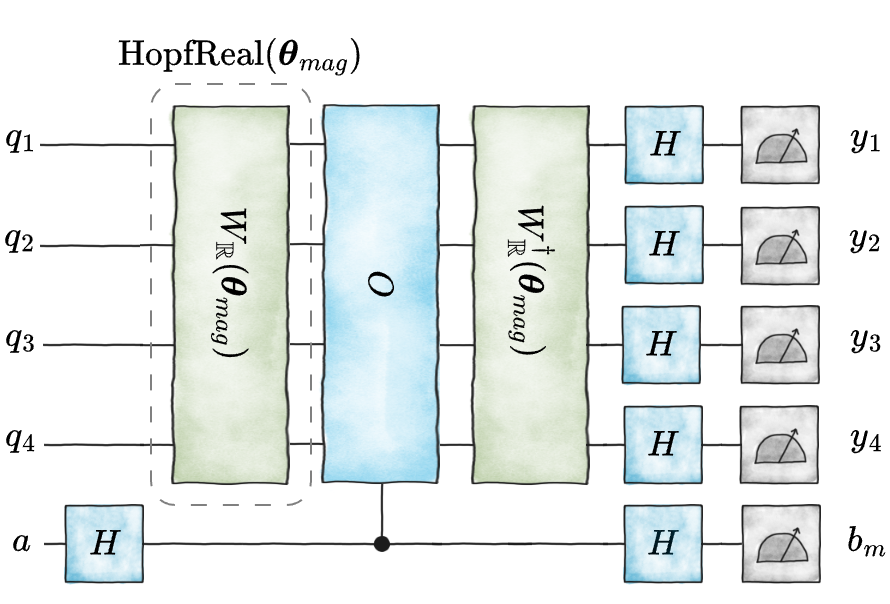}
    \captionof{figure}{Real-Hopf global-gradient circuit.  The forward bracket
      uses the state-column contract
      $\HopfReal(\boldsymbol\theta_{\mathrm{mag}})\ket{0000}=\WR(\boldsymbol\theta_{\mathrm{mag}})\ket{0000}$.
      The reverse block is $\WR^\dagger$.  The measured bits are $b_m$ and
      $\boldsymbol y=(y_1,y_2,y_3,y_4)$.}
    \label{fig:real-global-frame-protocol}
  \end{minipage}
\end{center}

Repeating this sole real-global template $S_{\mathrm{mag}}^{\mathbb R}$ times produces $S_{\mathrm{mag}}^{\mathbb R}$ joint five-bit outcomes, each reused for all fifteen markers. For node $5$, $\lambda(5)=0110$ and $\sqrt{g^{\mathbb R}_{5,5}}=\cos\theta_1\sin\theta_2$, so the main-text parity record becomes
\begin{equation*}
  Z_5^{\mathbb R} = 2\cos\theta_1\sin\theta_2 (-1)^{b_m+y_3+y_2}.
\end{equation*}
The other fourteen components are read from the same outcome using their markers in Table~\ref{tab:app-four-qubit-node-registry}.

For complex magnitudes, the initialized-state chain is the first line of Eq.~\eqref{eq:app-WC-frame-columns}. The same Walsh parity then decodes the phase-dressed magnitude frame.

\begin{center}
  \begin{minipage}{\textwidth}
    \centering
    \maybegraphic[scale=0.25]{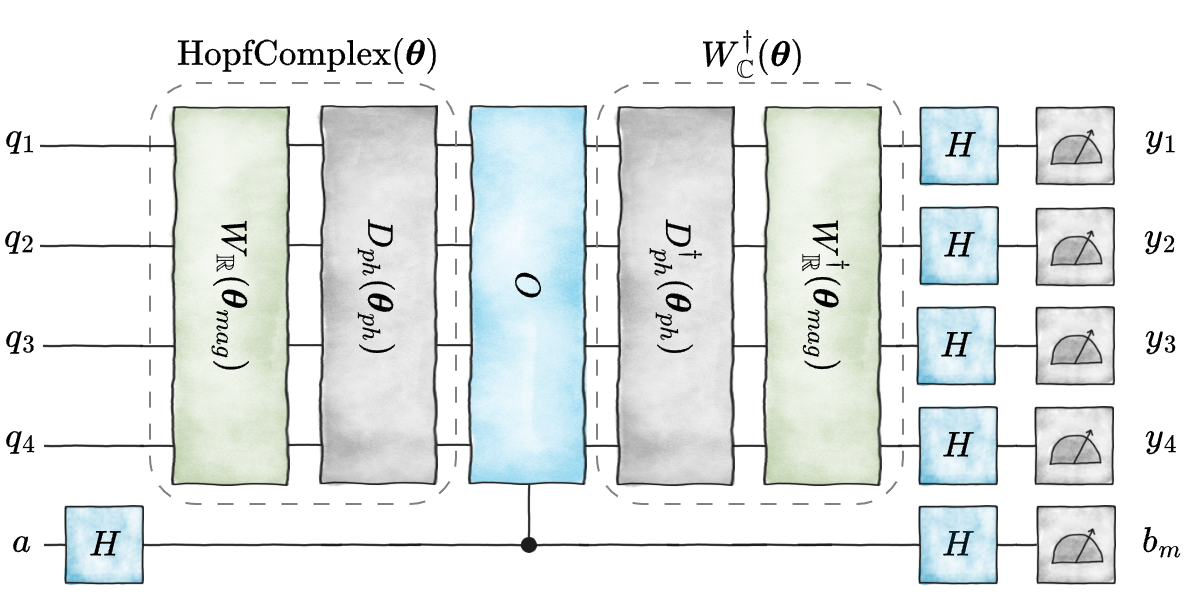}
    \captionof{figure}{Complex-Hopf magnitude circuit.  The forward bracket uses
      the state-column chain in Eq.~\eqref{eq:app-WC-frame-columns}.
      The reverse bracket uses the full-unitary identity
      $W_{\mathbb C}^{\dagger}=\WR^{\dagger}\Dph^{\dagger}$, so the physical order is
      $\Dph^{\dagger}$ followed by $\WR^{\dagger}$.  The addressed $R_{\mathbb C}$ compiler
      realizes the same decoder up to the compiler global phase in
      Eq.~\eqref{eq:app-WC-RC-compiler-identity}.}
    \label{fig:complex-magnitude-protocol}
  \end{minipage}
\end{center}

Repeating this magnitude template $S_{\mathrm{mag}}^{\mathbb C}$ times gives $S_{\mathrm{mag}}^{\mathbb C}$ joint five-bit all-$X$ outcomes, and every outcome is reused across all fifteen complex-magnitude markers.

The leaf phases use the direct ancilla-$Y$/system-$Z$ record of Eq.~\eqref{eq:complex-phase-estimator}.

\begin{center}
  \begin{minipage}{\textwidth}
    \centering
    \maybegraphic[scale=0.25]{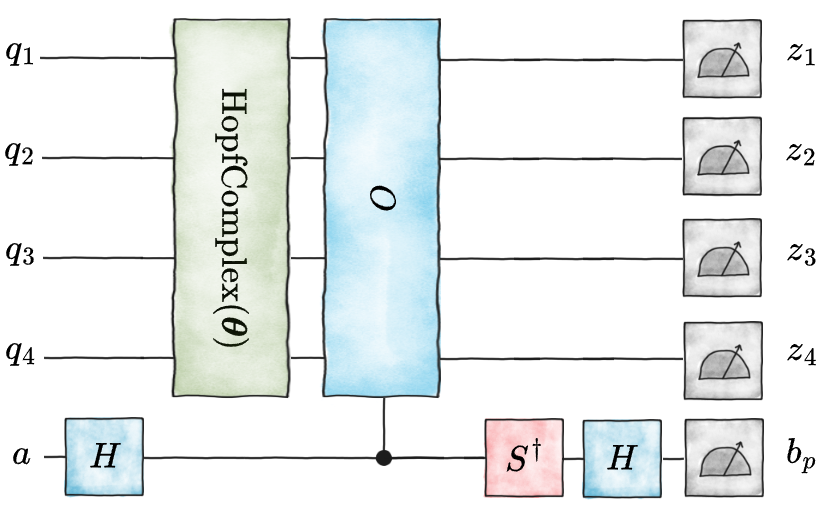}
    \captionof{figure}{Direct complex phase circuit.  Native
      $\HopfComplex(\boldsymbol\theta)$ prepares the state through the state-column
      contract in Eq.~\eqref{eq:app-WC-frame-columns}.  After the
      controlled reflection, the system is measured in $Z$ and the branch ancilla
      in $Y$, implemented by $S^\dagger$ followed by $H$.  This circuit uses the
      leaf-local phase record directly.}
    \label{fig:complex-phase-protocol}
  \end{minipage}
\end{center}

Repeating the phase template $S_{\mathrm{ph}}^{\mathbb C}$ times gives a signed histogram over the sixteen leaves; each five-bit outcome updates exactly one phase bin. Thus the three figures implement, respectively, all fifteen real magnitudes, all fifteen complex magnitudes, and all sixteen complex leaf phases without introducing any four-qubit-specific decoder formula.

\subsection{Checkpoint records and complex suffix $B_{2,\mathbb C}$}
\label{app:four-qubit-checkpoints}

The four real checkpoint settings are:

\begin{center}
  \begin{minipage}{0.9\textwidth}
    \centering
    \scriptsize
    \setlength{\tabcolsep}{5pt}
    \begin{tabular}{c c c c c}
      \toprule
      $d$ & Prefix in $Z$ & Target in $Y$ & Inverse suffix & Discard\\
      \midrule
      $0$ & none & $q_4$ & $U_1^\dagger U_2^\dagger U_3^\dagger$ & $q_3q_2q_1$\\
      $1$ & $q_4$ & $q_3$ & $U_2^\dagger U_3^\dagger$ & $q_2q_1$\\
      $2$ & $q_4q_3$ & $q_2$ & $U_3^\dagger$ & $q_1$\\
      $3$ & $q_4q_3q_2$ & $q_1$ & $I$ & none\\
      \bottomrule
    \end{tabular}
    \captionof{table}{All four real checkpoint settings.  Compiler charges are consolidated in Table~\ref{tab:app-four-qubit-resource-ledger}.}
    \label{tab:app-checkpoint-settings}
  \end{minipage}
\end{center}

At depth $2$, the measured prefix and coordinate index are
\begin{equation*}
  \widehat r=(z_4z_3)_2, \qquad j=4+\widehat r.
\end{equation*}
The target $q_2$ and branch ancilla are measured in $Y$, the prefix $q_4q_3$ in $Z$, and $q_1$ is discarded.

\begin{center}
  \begin{minipage}{\textwidth}
    \centering
    \maybegraphic[scale=0.25]{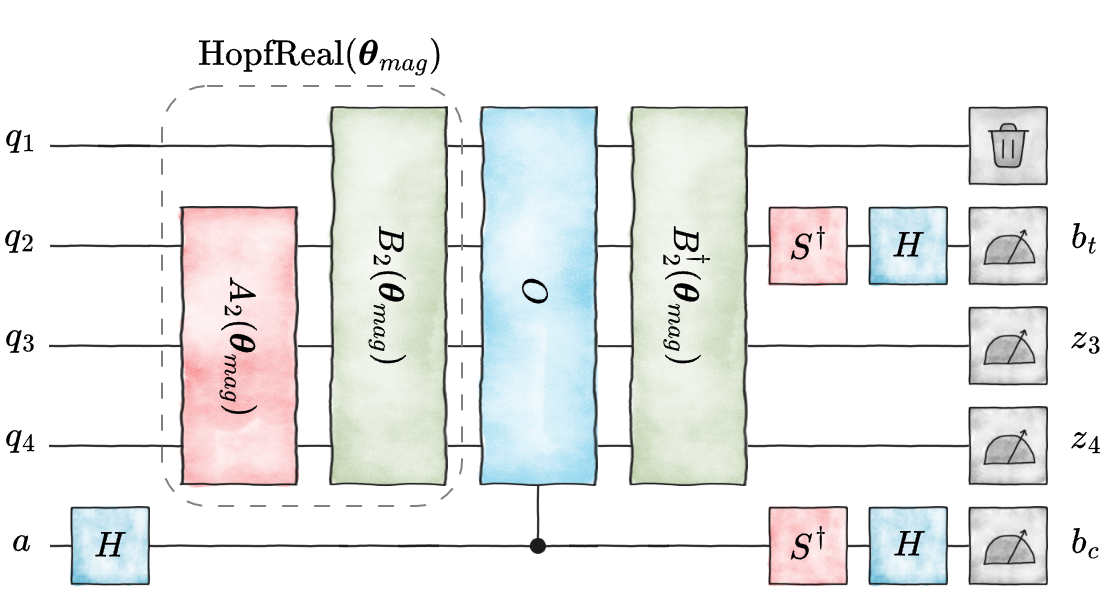}
    \captionof{figure}{Real-Hopf depth-$2$ checkpoint.  The forward bracket uses
      the state-column contract
      $\HopfReal(\boldsymbol\theta_{\mathrm{mag}})\ket{0000}=B_2A_2\ket{0000}$.
      After controlled-$O$, the reverse suffix is $B_2^\dagger=U_3^\dagger$.  One
      execution returns $-2(-1)^{b_c+b_t}\mathbf e_{\widehat r}$ for nodes
      $4,5,6,7$.}
    \label{fig:real-checkpoint-depth2}
  \end{minipage}
\end{center}

For node $5$, corresponding to $r=1$,
\begin{equation*}
  Q_{2,1} =\ket{01}\!\bra{01}_{q_4q_3}\otimes Y_{q_2}\otimes I_{q_1},
\end{equation*}
and with $\ket{\alpha_2}:=A_2\ket{0000}$ and $\ket{\beta_2^{\mathbb R}}:=B_2^\dagger O\ket{\psi^{\mathbb R}}$,
\begin{equation*}
  \partial_{\theta_5}\EO = 2\Imag\bra{\beta_2^{\mathbb R}}Q_{2,1}\ket{\alpha_2}.
\end{equation*}

\subsubsection{Complex depth-$2$ active-interface construction}
\label{app:four-qubit-complex-checkpoint}

The reached-interface projector is
\begin{equation*}
  P_2:=I_{q_4q_3q_2}\otimes\ket0\!\bra0_{q_1}.
\end{equation*}
Define the eight-sector completion compactly by
\begin{equation*}
  \boxed{ B_{2,\mathbb C}(\boldsymbol\theta) := \sum_{p=0}^7 \ket p\!\bra p_{q_4q_3q_2} \otimes R_{\mathbb C}^{(q_1)} \!\left(\theta_{8+p},\phi_{2p},\phi_{2p+1}\right). }
\end{equation*}
Its first column in each prefix sector satisfies
\begin{equation*}
  B_{2,\mathbb C}\ket{p,0} = e^{\iu\phi_{2p}}\cos\theta_{8+p}\ket{p,0} + e^{\iu\phi_{2p+1}}\sin\theta_{8+p}\ket{p,1} = \Dph B_2\ket{p,0},
\end{equation*}
so
\begin{equation}
  \boxed{B_{2,\mathbb C}P_2=\Dph B_2P_2.}
  \label{eq:app-B2C-interface-identity}
\end{equation}
The unused second column of each $R_{\mathbb C}$ gate completes the unitary outside the reached interface. With $\ket{\beta_2^{\mathbb C}}:=B_{2,\mathbb C}^\dagger O\ket{\psi^{\mathbb C}}$, the four transported derivatives are
\begin{equation}
  \partial_{\theta_{4+r}}\EO = 2\Imag\bra{\beta_2^{\mathbb C}}Q_{2,r}\ket{\alpha_2}, \qquad r=0,1,2,3.
  \label{eq:app-complex-depth2-checkpoint-identity}
\end{equation}

\begin{center}
  \begin{minipage}{\textwidth}
    \centering
    \maybegraphic[scale=0.25]{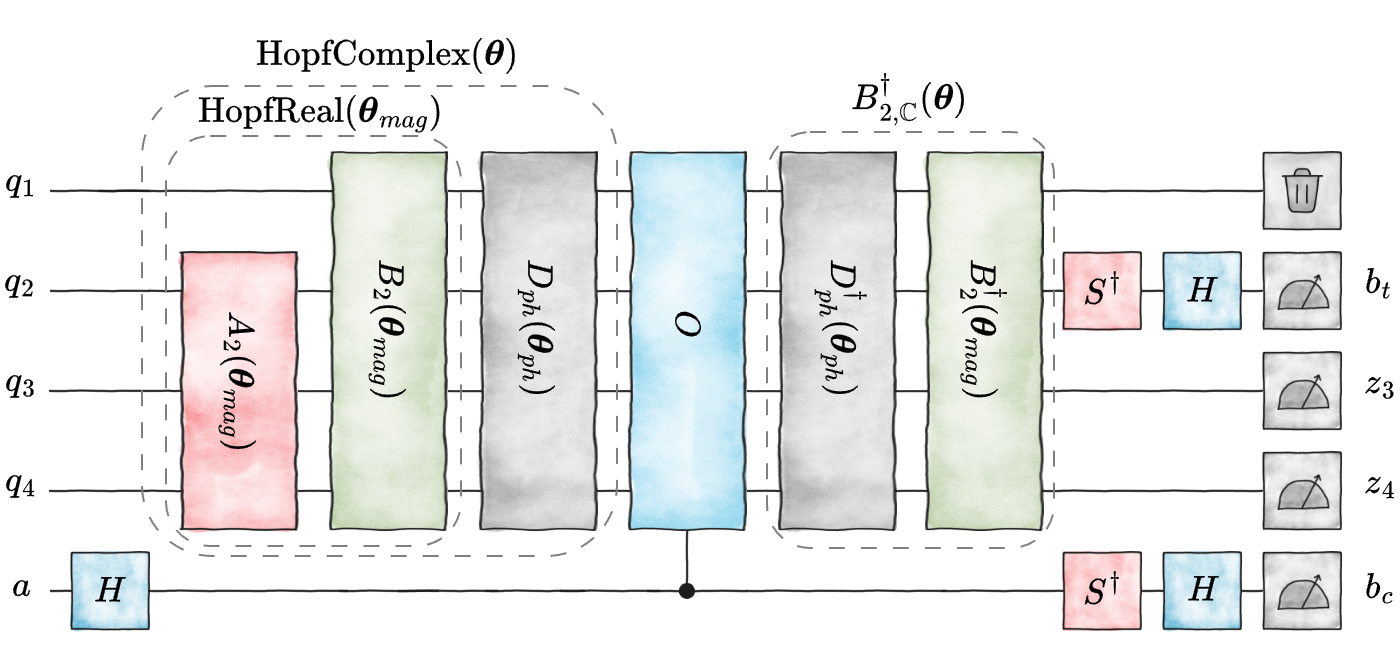}
    \captionof{figure}{Complex-Hopf magnitude checkpoint at depth $2$.  The
      inner and outer forward brackets use the real and complex state-column
      contracts.  The reverse bracket uses the active-interface identity
      $B_{2,\mathbb C}P_2=\Dph B_2P_2$, which gives the same checkpoint correlators
      as the separated physical sequence $\Dph^\dagger$ followed by $B_2^\dagger$.
      The readout measures $q_4q_3$ in $Z$ and $q_2$ together with the ancilla in
      $Y$, while $q_1$ is discarded because it belongs to the unused lower suffix.}
    \label{fig:complex-checkpoint-depth2}
  \end{minipage}
\end{center}

In either depth-$2$ circuit, one execution gives
\begin{equation*}
  Z_2^{\mathrm{chk}} =-2(-1)^{b_c+b_t}\mathbf e_{\widehat r}, \qquad \E Z_2^{\mathrm{chk}} = \left( \partial_{\theta_4}\EO, \partial_{\theta_5}\EO, \partial_{\theta_6}\EO, \partial_{\theta_7}\EO \right)^\mathsf T, \qquad \|Z_2^{\mathrm{chk}}\|_2=2.
\end{equation*}
Repeating the selected depth-$2$ template and averaging this one-hot record recovers all four displayed derivatives. Recovering another depth requires its own suffix and measurement target from Table~\ref{tab:app-checkpoint-settings}.

\subsection{Four-qubit assigned Hopf CNOT ledger}
\label{app:four-qubit-ledger}

Under the adopted direct-angle ledger inherited from the first paper, a three-controlled $R_y$ or $R_{\mathbb C}$ gate costs $14$ CNOTs. The native schedule contains one uncontrolled, four one-controlled, six two-controlled, and four three-controlled gates, giving charge $100$. The depth layers cost $0,4,24,112$, so $C_{\mathrm H}(\Uchk)=140$ and $C_{\mathrm H}(B_2)=112$. The addressed frames use fifteen three-controlled gates, whereas $B_{2,\mathbb C}$ uses the eight sectors of the selected depth-$2$ suffix.

\begin{center}
  \begin{minipage}{\textwidth}
    \centering
    \scriptsize
    \setlength{\tabcolsep}{4.5pt}
    \begin{tabular}{l c p{8.7cm}}
      \toprule
      Item & Charge & Role, contract, or implemented blocks\\
      \midrule
      $\HopfReal$ & $100$ & Native real state preparation.\\
      $\HopfComplex$ & $100$ & Native complex state preparation with eight final-layer $R_{\mathbb C}$ gates.\\
      $\Uchk$ & $140$ & Depth completion and designated separated scalar/checkpoint forward program.\\
      $B_2$ & $112$ & Real inverse suffix at depth $2$.\\
      $\WR$ & $210$ & Complete real state-and-tangent frame.\\
      $W_{\mathbb C}$ & $210$ & Complete complex magnitude frame, compiled up to one compiler global phase.\\
      $B_{2,\mathbb C}$ & $112$ & Complex suffix satisfying $B_{2,\mathbb C}P_2=\Dph B_2P_2$.\\
      \midrule
      Real global record & $310$ & $\HopfReal$ followed by $\WR^\dagger$.\\
      Complex magnitude record & $310$ & $\HopfComplex$ followed by $W_{\mathbb C}^\dagger$.\\
      Direct complex phase record & $100$ & $\HopfComplex$ with no reverse frame.\\
      Real depth-$2$ checkpoint & $212$ & $\HopfReal$ followed by $B_2^\dagger$.\\
      Complex depth-$2$ checkpoint & $212$ & $\HopfComplex$ followed by $B_{2,\mathbb C}^\dagger$.\\
      \bottomrule
    \end{tabular}
    \captionof{table}{Four-qubit compiler and record charges under the same
      direct-angle $C_{\mathrm H}$ functional used in the theorem; the controlled
      observable and
      readout are excluded by definition.  The theorem uses the designated separated
      schedules, while this table evaluates native-forward and integrated reverse
      alternatives whose finite constant ratios need not be $2$ and $1$.}
    \label{tab:app-four-qubit-resource-ledger}
  \end{minipage}
\end{center}

The $310$ global charge contributes to all fifteen magnitude coordinates; the $212$ checkpoint charge contributes to the four active depth-$2$ coordinates. Native-forward records from all four real checkpoint families cost $240+236+212+100=788$ in the same Hopf-only ledger. The finite comparison displays the same global-versus-local protocol choice: checkpointing is cheaper for an active layer, while the complete frame reuses one preparation, oracle call, and outcome across every depth. These native and integrated finite ratios are illustrative; the theorem instead uses the constants of the designated separated schedules.

\FloatBarrier

\section{Technical foundations, proofs, and compiler derivations}
\label{app:all-proofs}

This appendix develops the proofs from the common interference identities to the complete-gradient theorem.

\subsection{Controlled-interference identities}
\label{app:generic-reverse}

Let $W\ket0^{\otimes n}=\ket\psi$ and let an admissible forward circuit prepare the same state. Starting from $\ket+_a\ket0^{\otimes n}$, the forward block, $\ctrl(O)$, and $W^\dagger$ produce Eq.~\eqref{eq:generic-frame-interference-state}. Since $X(k)\ket0^{\otimes n}=\ket{b_k}$,
\begin{equation*}
  \bra{\Omega_W}X_a\otimes X(k)\ket{\Omega_W} = \frac12\left( \bra0X(k)W^\dagger O\ket\psi + \bra\psi O W X(k)\ket0 \right) = \Real\bra{b_k}W^\dagger O\ket\psi.
\end{equation*}
The measured eigenvalue of $X_a\otimes X(k)$ is $(-1)^{b+k\cdot\boldsymbol y}$, proving Eq.~\eqref{eq:frame-parity-identity}.

For normalized branches $\ket\alpha,\ket\beta$ and a Hermitian system operator $A$, set $\ket\Omega=(\ket0\ket\alpha+\ket1\ket\beta)/\sqrt2$. Using $\bra0Y\ket1=-\iu$ and $\bra1Y\ket0=\iu$ gives
\begin{equation}
  \bra\Omega Y_a\otimes A\ket\Omega = \Imag\bra\alpha A\ket\beta = -\Imag\bra\beta A\ket\alpha.
  \label{eq:generic-Y-quadrature}
\end{equation}
The direct phase record uses the first form and the checkpoint record the second.

\subsection{Balanced Hopf differential frame}
\label{app:proof-real-geometry}

For a tree node $j$, let $\mathcal H_j$ be the span of its descendant basis leaves and let $\mathcal I(j)$ be the internal nodes in that subtree.

\begin{lemma}[Subtree frame]
  \label{lem:subtree-basis}
  The set $\mathcal B_j:=\{\ket{\sigma_j}\}\cup \{\ket{e_k^{\mathbb R}}:k\in\mathcal I(j)\}$ is an orthonormal basis of $\mathcal H_j$.
\end{lemma}

\begin{proof}
  At a leaf, $\mathcal B_j=\{\ket{b_\ell}\}$. Assume the statement for both children of an internal node $j$. Their leaf spaces are orthogonal. In the direct sum of the two child bases, replace the pair $\ket{\sigma_{2j}},\ket{\sigma_{2j+1}}$ by
  \begin{equation*}
    \begin{pmatrix}
      \ket{\sigma_j} \\[2pt]
      \ket{e_j^{\mathbb R}}
    \end{pmatrix}
    =
    \begin{pmatrix}
      \cos\theta_j & \sin\theta_j \\
      -\sin\theta_j & \cos\theta_j
    \end{pmatrix}
    \begin{pmatrix}
      \ket{\sigma_{2j}} \\[2pt]
      \ket{\sigma_{2j+1}}
    \end{pmatrix}.
  \end{equation*}
  This orthogonal change of basis leaves every strict-descendant complement unchanged, completing the induction.
\end{proof}

At the root, Lemma~\ref{lem:subtree-basis} gives Eq.~\eqref{eq:root-moving-frame}. For an internal node $j$, expansion along the unique root-to-$j$ path gives
\begin{equation}
  \ket{\psi^{\mathbb R}} = \ket{\psi_{\neg j}} + \sqrt{g^{\mathbb R}_{j,j}}\ket{\sigma_j}, \qquad \ket{\psi_{\neg j}}\in\mathcal H_j^\perp.
  \label{eq:subtree-state-decomposition}
\end{equation}
The first term and incoming coefficient are independent of $\theta_j$. The coefficient is the nonnegative product of ancestor sine and cosine factors; its square obeys Eq.~\eqref{eq:incoming-metric-recursion}. Since $\partial_{\theta_j}\ket{\sigma_j}=\ket{e_j^{\mathbb R}}$, Eq.~\eqref{eq:hopf-derivative-factorization} follows. Orthonormality gives Eq.~\eqref{eq:hopf-diagonal-metric}. For $g^{\mathbb R}_{j,j}=0$, the coordinate derivative vanishes by Eq.~\eqref{eq:subtree-state-decomposition}; Eq.~\eqref{eq:real-frame-column-recursion} still defines a normalized completion column.

For a Hermitian observable,
\begin{equation*}
  \partial_{\theta_j}\EO = \bra{\partial_{\theta_j}\psi}O\ket\psi + \bra\psi O\ket{\partial_{\theta_j}\psi} = 2\sqrt{g^{\mathbb R}_{j,j}}\Real\bra{e_j^{\mathbb R}}O\ket{\psi^{\mathbb R}},
\end{equation*}
which is Eq.~\eqref{eq:hopf-transition-gradient}.

\subsection{Depth completion and forward-state substitution}
\label{app:proof-depth-completion}

Define real path coefficients recursively by
\begin{equation}
  a_{0,0}=1, \qquad a_{d+1,2r}=a_{d,r}\cos\theta_{2^d+r}, \qquad a_{d+1,2r+1}=a_{d,r}\sin\theta_{2^d+r}.
  \label{eq:depth-path-coefficients}
\end{equation}
After the first $d+1$ depth layers,
\begin{equation}
  U_d\cdots U_0\ket0^{\otimes n} = \sum_{r=0}^{2^{d+1}-1} a_{d+1,r} \ket r_{q_n\cdots q_{n-d}} \ket0_{q_{n-d-1}\cdots q_1}.
  \label{eq:depth-prefix-state}
\end{equation}
For $d=0$, this is the rotation in Eq.~\eqref{eq:paper-Ry-convention} on $q_n$. Assume the statement through depth $d-1$. In the component with prefix $r$, exactly one projector in $U_d$ is active, and it rotates the fresh target $q_{n-d}$ from $\ket0$ to $\cos\theta_{2^d+r}\ket0+\sin\theta_{2^d+r}\ket1$. This gives the two recursions in Eq.~\eqref{eq:depth-path-coefficients}, so Eq.~\eqref{eq:depth-prefix-state} follows by induction.

At $d=n-1$, the coefficients are the real Hopf leaf amplitudes. Hence $\Uchk\ket0^{\otimes n}=\ket{\psi^{\mathbb R}}$, proving Eq.~\eqref{eq:depth-hopf-state}; applying $\Dph$ gives $\Dph\Uchk\ket0^{\otimes n}=\ket{\psi^{\mathbb C}}$, which, together with the native complex state column, proves Eq.~\eqref{eq:designated-complex-preparation}.

For forward substitution, let $U\ket0^{\otimes n}=V\ket0^{\otimes n}=\ket\psi$. Either forward block gives the same complete joint state, so identical subsequent unitaries and measurements give the same final state and output distribution. This establishes state-column substitution for the real-global, complex-global, direct-phase, and checkpoint records.

The forward and reverse circuits may therefore use different unitary completions of the same prepared state. This is a correctness statement: the assigned-work comparison additionally fixes the direct-angle forward and reverse schedules of Sec.~\ref{subsec:assigned-ledger}. A state-equivalent recompilation with different elementary physical angles requires its own ledger.

\subsection{Addressed real-frame columns}
\label{app:proof-real-frame-columns}

Every nonzero integer $k<N$ has a unique factorization $k=(2r+1)2^t$ with $t\geq0$. Setting $d=n-t-1$ gives $0\leq r<2^d$ and $k=(2r+1)2^{n-d-1}=\lambda(2^d+r)$, so the marker map is a bijection.

Let $W_j$ be the product of addressed rotations in the subtree rooted at $j$, with $G_j$ applied before its strict descendants. The child products act on orthogonal leaf spaces, and $W_j=(W_{2j}\oplus W_{2j+1})G_j$. At a leaf, the anchor column is already the leaf state. Assuming the child anchor columns map to the child subtree states, Eqs.~\eqref{eq:frame-gate-anchor} and~\eqref{eq:frame-gate-marker} give
\begin{align*}
  W_j\ket{b_{\ell_0(j)}} &= \cos\theta_j\ket{\sigma_{2j}} + \sin\theta_j\ket{\sigma_{2j+1}} = \ket{\sigma_j},\\
  W_j\ket{b_{\lambda(j)}} &= -\sin\theta_j\ket{\sigma_{2j}} + \cos\theta_j\ket{\sigma_{2j+1}} = \ket{e_j^{\mathbb R}}.
\end{align*}
The same induction preserves all strict-descendant marker columns. At the root, $\ell_0(1)=0$, proving Eq.~\eqref{eq:frame-column-identities}.

\subsection{Global real readout and Walsh decoding}
\label{app:proof-real-readout}

For any admissible real forward preparation, the state-column contract gives $\WR^\dagger U_{\mathrm{fwd}}^{\mathbb R}\ket0=\ket0$ and $\WR^\dagger O U_{\mathrm{fwd}}^{\mathbb R}\ket0 =\WR^\dagger O\WR\ket0$. Hence
\begin{align*}
  \ket{\Omega_{\mathbb R}} &= \frac{ \ket0_a\ket0^{\otimes n} + \ket1_a\WR^\dagger O\WR\ket0^{\otimes n} }{\sqrt2},\\
  \E\!\left[(-1)^{b_m+\lambda(j)\cdot\boldsymbol y}\right] &= \Real\bra{e_j^{\mathbb R}}O\ket{\psi^{\mathbb R}},
\end{align*}
where the second line uses Eq.~\eqref{eq:frame-parity-identity} at $k=\lambda(j)$ and Eq.~\eqref{eq:frame-column-identities}. Multiplication by $2\sqrt{g^{\mathbb R}_{j,j}}$ proves Eqs.~\eqref{eq:real-hopf-estimator} and \eqref{eq:real-hopf-estimator-properties}. Averaging independent copies gives Eq.~\eqref{eq:real-hopf-direct-average}; the same argument after phase dressing gives Eq.~\eqref{eq:complex-magnitude-direct-average}. A zero metric factor makes the record zero for every outcome.

The metric mass is conserved across each tree depth, since the root weight is one and every parent splits by $\cos^2\theta_j+\sin^2\theta_j=1$. Therefore
\begin{align*}
  \sum_{r=0}^{2^d-1}g^{\mathbb R}_{2^d+r,2^d+r}&=1,\\
  \|Z_d^{\mathrm{fr}}\|_2^2 &=4\sum_{r=0}^{2^d-1}g^{\mathbb R}_{2^d+r,2^d+r}=4,
\end{align*}
which is Eq.~\eqref{eq:global-depth-fixed-norm}.

Finally,
\begin{equation*}
  \frac1S\sum_{t=1}^S (-1)^{b_{m,t}+k\cdot\boldsymbol y_t} = \frac1S\sum_{\boldsymbol y} \left[ \sum_{t:\boldsymbol y_t=\boldsymbol y}(-1)^{b_{m,t}} \right] (-1)^{k\cdot\boldsymbol y} = \frac1S\sum_{\boldsymbol y}h(\boldsymbol y)(-1)^{k\cdot\boldsymbol y}.
\end{equation*}
This is Eq.~\eqref{eq:fwht-moments}. The fast Walsh--Hadamard transform uses $\bigO(N\log N)$ additions and subtractions~\cite{fino1976walsh}; one tree pass computes all metric factors in $\bigO(N)$ further operations.

\subsection{Complex magnitude and direct phase identities}
\label{app:proof-complex}

Since $\Dph$ is independent of every magnitude angle and is unitary,
\begin{align*}
  \partial_{\theta_j}\ket{\psi^{\mathbb C}} &= \Dph\partial_{\theta_j}\ket{\psi_{\mathrm{mag}}} = \sqrt{g^{\mathbb C}_{j,j}}\ket{e_j^{\mathbb C}},\\
  W_{\mathbb C}\ket0&=\ket{\psi^{\mathbb C}}, & W_{\mathbb C}\ket{b_{\lambda(j)}}&=\ket{e_j^{\mathbb C}}.
\end{align*}
The first line proves Eq.~\eqref{eq:complex-magnitude-derivative}; the second uses $W_{\mathbb C}=\Dph\WR$ and the real frame columns. For any admissible complex forward preparation,
\begin{align*}
  W_{\mathbb C}^\dagger U_{\mathrm{fwd}}^{\mathbb C}\ket0 &=\ket0,\\
  W_{\mathbb C}^\dagger O U_{\mathrm{fwd}}^{\mathbb C}\ket0 &=W_{\mathbb C}^\dagger O\ket{\psi^{\mathbb C}},\\
  \E(-1)^{b_m+\lambda(j)\cdot\boldsymbol y} &= \Real\bra{e_j^{\mathbb C}}O\ket{\psi^{\mathbb C}}.
\end{align*}
Thus Eqs.~\eqref{eq:complex-magnitude-estimator} and \eqref{eq:complex-magnitude-estimator-bound} follow. Phase dressing preserves the metric factors, so each depth record has norm two.

For a phase coordinate, direct differentiation and the controlled-interference identity give
\begin{align*}
  \partial_{\theta_{N+\ell}}\ket{\psi^{\mathbb C}} &= \iu\Pi_\ell\ket{\psi^{\mathbb C}} = \iu x_\ell\ket{b_\ell},\\
  \partial_{\theta_{N+\ell}}\EO &= 2\Imag\bra{\psi^{\mathbb C}}\Pi_\ell O\ket{\psi^{\mathbb C}},\\
  \ket{\Omega_{\mathrm{ph}}} &= \frac{\ket0_a\ket{\psi^{\mathbb C}}+ \ket1_aO\ket{\psi^{\mathbb C}}}{\sqrt2},\\
  \bra{\Omega_{\mathrm{ph}}}Y_a\otimes\Pi_\ell\ket{\Omega_{\mathrm{ph}}} &= \Imag\bra{\psi^{\mathbb C}}\Pi_\ell O\ket{\psi^{\mathbb C}}.
\end{align*}
The first line is Eq.~\eqref{eq:complex-phase-raw-tangent}; the last follows from Eq.~\eqref{eq:generic-Y-quadrature} with $A=\Pi_\ell$. The measured eigenvalue is $(-1)^{b_p}\indicator{\widehat\ell=\ell}$, so multiplication by two proves Eqs.~\eqref{eq:complex-phase-estimator}-- \eqref{eq:complex-phase-onehot-vector}.

The phase tangents also satisfy $\Real\langle\partial_{\theta_{N+\ell}}\psi^{\mathbb C} |\partial_{\theta_{N+m}}\psi^{\mathbb C}\rangle =p_\ell\delta_{\ell m}$, with $p_\ell:=|x_\ell|^2$. Hence $G_{\mathrm{ph}}^{\mathrm{sph}}=\operatorname{diag}(\boldsymbol p)$; subtracting the state-parallel part gives the unused projective block $G_{\mathrm{ph}}^{\mathrm{FS}}=\operatorname{diag}(\boldsymbol p) -\boldsymbol p\boldsymbol p^{\mathsf T}$.

For the integrated four-qubit $R_{\mathbb C}$ compiler in Appendix~\ref{app:four-qubit-RC-compiler}, let $D_{\widehat\phi}:=\sum_{\ell=0}^{15}e^{\iu\widehat\phi_\ell} \ket{b_\ell}\!\bra{b_\ell}$. Suppose the two child-anchor columns have already expanded to
\begin{align*}
  \ket{\widetilde\sigma_{2j}} &=e^{-\iu\Phi_{2j}}D_{\widehat\phi}\ket{\sigma_{2j}},\\
  \ket{\widetilde\sigma_{2j+1}} &=e^{-\iu\Phi_{2j+1}}D_{\widehat\phi}\ket{\sigma_{2j+1}}.
\end{align*}
At a leaf, $\ket{b_\ell}=e^{-\iu\Phi_{16+\ell}}D_{\widehat\phi}\ket{b_\ell}$, so the induction begins. The first and second columns of the addressed gate in Eq.~\eqref{eq:app-WC-RC-gate} produce, respectively,
\begin{align*}
  e^{-\iu\Phi_{2j+1}}\cos\theta_j\ket{\widetilde\sigma_{2j}} + e^{-\iu\Phi_{2j}}\sin\theta_j\ket{\widetilde\sigma_{2j+1}} &= e^{-\iu\Phi_j}D_{\widehat\phi}\ket{\sigma_j},\\
  -e^{\iu\Phi_{2j}}\sin\theta_j\ket{\widetilde\sigma_{2j}} + e^{\iu\Phi_{2j+1}}\cos\theta_j\ket{\widetilde\sigma_{2j+1}} &= D_{\widehat\phi}\ket{e_j^{\mathbb R}}.
\end{align*}
At the root, $\Phi_1=0$. The state column and all fifteen marker columns are therefore those of $D_{\widehat\phi}\WR$. These are all sixteen columns, so
\begin{equation*}
  \widetilde W_{\mathbb C}^{(R_{\mathbb C})} = D_{\widehat\phi}\WR = e^{-\iu\bar\phi}\Dph\WR = e^{-\iu\bar\phi}W_{\mathbb C},
\end{equation*}
proving Eq.~\eqref{eq:app-WC-RC-compiler-identity}. The common phase on the complete decoded state leaves measurement probabilities unchanged.

\subsection{Checkpoint transport and active-interface substitution}
\label{app:proof-real-checkpoint}

The generators within one depth commute because their prefix projectors are orthogonal. For $j=2^d+r$, using $\Uchk=B_dA_d$ gives
\begin{align*}
  \partial_{\theta_j}U_d &=-\iu Q_{d,r}U_d,\\
  \partial_{\theta_j}\ket{\psi^{\mathbb R}} &=-\iu B_dQ_{d,r}\ket{\alpha_d},\\
  \partial_{\theta_j}\EO &= 2\Imag\bra{\beta_d^{\mathbb R}}Q_{d,r}\ket{\alpha_d}, \qquad \ket{\beta_d^{\mathbb R}}:=B_d^\dagger O\ket{\psi^{\mathbb R}}.
\end{align*}
The last line is Eq.~\eqref{eq:checkpoint-adjoint-identity}. For the state in Eq.~\eqref{eq:checkpoint-interference-state}, Eq.~\eqref{eq:generic-Y-quadrature} gives
\begin{align*}
  \bra{\Omega_d^{\mathbb R}}Y_a\otimes Q_{d,r}\ket{\Omega_d^{\mathbb R}} &= -\Imag\bra{\beta_d^{\mathbb R}}Q_{d,r}\ket{\alpha_d},\\
  \E\!\left[(-1)^{b_c+b_t}\indicator{\widehat r=r}\right] &= \bra{\Omega_d^{\mathbb R}}Y_a\otimes Q_{d,r}\ket{\Omega_d^{\mathbb R}}.
\end{align*}
Measuring the prefix in $Z$ realizes $\ket r\!\bra r$, measuring the target in $Y$ realizes the Pauli $Y$, and the lower suffix may be discarded because $Q_{d,r}$ acts there as the identity. Multiplication by $-2$ proves Eqs.~\eqref{eq:checkpoint-onehot-vector} and \eqref{eq:checkpoint-depth-gradient}. Exactly one prefix is observed, so Eq.~\eqref{eq:checkpoint-fixed-norm} follows.

Equation~\eqref{eq:depth-prefix-state} shows that $\ket{\alpha_d}$ has a zero lower suffix. Since $Q_{d,r}$ acts there as the identity, both $\ket{\alpha_d}$ and $Q_{d,r}\ket{\alpha_d}$ are fixed by the active-interface projector $P_d$. For the complex state, differentiation and the interface contract give
\begin{align}
  \partial_{\theta_j}\ket{\psi^{\mathbb C}} &= -\iu\Dph B_dQ_{d,r}\ket{\alpha_d} = -\iu B_{d,\mathbb C}Q_{d,r}\ket{\alpha_d},
  \label{eq:proof-complex-checkpoint-derivative}\\
  \ket{\psi^{\mathbb C}} &= \Dph B_d\ket{\alpha_d} = B_{d,\mathbb C}\ket{\alpha_d}. \nonumber
\end{align}
Thus the reference branch decodes to $\ket{\alpha_d}$. With $\ket{\beta_d^{\mathbb C}}=B_{d,\mathbb C}^\dagger O \ket{\psi^{\mathbb C}}$, the same derivation applies, and its correlators depend only on $B_{d,\mathbb C}P_d$. The separated choice $B_{d,\mathbb C}=\Dph B_d$ gives Eq.~\eqref{eq:complex-checkpoint-circuit-order}. The explicit four-qubit first-column calculation in Appendix~\ref{app:four-qubit-complex-checkpoint} proves Eq.~\eqref{eq:app-B2C-interface-identity}; substituting it into Eq.~\eqref{eq:proof-complex-checkpoint-derivative} gives Eq.~\eqref{eq:app-complex-depth2-checkpoint-identity}. Generators within one depth commute, while different depths retain their physical circuit order.

\subsection{Fixed-norm concentration and complete-chart allocations}
\label{app:proof-resource-criterion}

Let $\widehat\mu=S^{-1}\sum_tZ_t$ and $\mu=\E Z_t$. Using $\|v\|_\infty\leq\|v\|_2$, Jensen's inequality, and independence gives
\begin{align*}
  \E\|Z_1-\mu\|_2^2 &= \E\|Z_1\|_2^2-\|\mu\|_2^2 \leq B^2,\\
  \E\|\widehat\mu-\mu\|_\infty &\leq \sqrt{\E\|\widehat\mu-\mu\|_2^2} \leq \frac{B}{\sqrt S}.
\end{align*}
Replacing one sample changes $\|\widehat\mu-\mu\|_\infty$ by at most $2B/S$. McDiarmid's inequality~\cite{mcdiarmid1989bounded} therefore gives
\begin{equation*}
  \Prob\!\left[ \|\widehat\mu-\mu\|_\infty - \E\|\widehat\mu-\mu\|_\infty >u \right] \leq \exp\!\left(-\frac{Su^2}{2B^2}\right).
\end{equation*}
Setting $u=B\sqrt{2\log(1/\eta)/S}$ proves Eq.~\eqref{eq:fixed-norm-tail}.

Every magnitude-depth, checkpoint, and direct phase record has norm two, so $S_\star(\eta)$ controls the $\ell_\infty$ error of the entire record vector. At depth $d$ the vector has $2^d$ coordinates, but its norm remains two; therefore the block width does not enter the count and no coordinate-wise union bound is needed.

There are $n$ magnitude depths. Assigning failure probability $\delta/n$ to each real depth, or $\delta/(2n)$ to each complex magnitude depth while reserving $\delta/2$ for the direct phase family, gives
\begin{equation*}
  S_\star(\delta/n),\ S_\star(\delta/(2n)) = \bigO\!\left( \frac{1+\log(n/\delta)}{\varepsilon^2} \right), \qquad S_\star(\delta/2) = \bigO\!\left( \frac{1+\log(1/\delta)}{\varepsilon^2} \right).
\end{equation*}
At fixed $\varepsilon$ and $\delta$, one magnitude stream costs $\bigO(\log n)$ executions and the phase stream costs $\bigO(1)$. One global outcome supplies all $n$ magnitude-depth records; a complete checkpoint schedule uses $n$ depth-specific streams. Thus
\begin{align*}
  S_{\mathrm{glob}}^{\mathbb K} &=\bigO(\log n),\\
  S_{\mathrm{chk}}^{\mathbb K} &=\bigO(n\log n),
\end{align*}
with the $\bigO(1)$ phase term included only in the complex chart. The union bound does not require the depth records decoded from one global outcome to be independent. These counts reproduce Eqs.~\eqref{eq:complete-chart-shot-allocations} and \eqref{eq:global-checkpoint-shot-scaling}.

\medskip
\noindent\textbf{Complete-vector Euclidean accuracy.}
Concatenate the $n$ magnitude-depth records from one global outcome as
$Z_{\mathrm{mag}}:=Z_0^{\mathrm{fr}}\oplus\cdots\oplus
  Z_{n-1}^{\mathrm{fr}}$, and set
$\mu_{\mathrm{mag}}:=\E Z_{\mathrm{mag}}$ and
$\widehat\mu_{\mathrm{mag}}:=S^{-1}\sum_{t=1}^{S}Z_{\mathrm{mag},t}$.
Since every block has norm two,
$\|Z_{\mathrm{mag}}\|_2=2\sqrt n$.  Applying the same bounded-difference
argument directly to the Euclidean norm gives
\begin{equation}
  \Prob\!\left[ \left\|\widehat\mu_{\mathrm{mag}}-\mu_{\mathrm{mag}}\right\|_2 > \frac{2\sqrt n\{1+\sqrt{2\log(1/\delta)}\}}{\sqrt S} \right] \leq\delta.
  \label{eq:complete-magnitude-l2-tail}
\end{equation}
A sufficient allocation for fixed complete-magnitude $\ell_2$ accuracy $\varepsilon_2$ is
\begin{equation*}
  S = \bigO\!\left( \frac{n\{1+\log(1/\delta)\}}{\varepsilon_2^2} \right) = \bigO\!\left( \frac{\log M\{1+\log(1/\delta)\}}{\varepsilon_2^2} \right).
\end{equation*}
For the complex chart, the separate norm-two phase stream adds only a constant-order term under any fixed split of Euclidean error and confidence.

\noindent\textbf{Relative, directional, and metric-conditioned accuracy.}
Let $\widehat{\boldsymbol v}=\boldsymbol v+\boldsymbol e$, where
$\boldsymbol v$ is the complete real gradient or complex magnitude block,
$\mathcal G:=\|\boldsymbol v\|_2$, and
$\|\boldsymbol e\|_2\leq\xi<\mathcal G$.  Then
\begin{equation*}
  \boldsymbol v^{\mathsf T}\widehat{\boldsymbol v} \geq\mathcal G(\mathcal G-\xi)>0, \qquad \left\|\frac{\widehat{\boldsymbol v}}{\|\widehat{\boldsymbol v}\|_2} -\frac{\boldsymbol v}{\mathcal G}\right\|_2 \leq\frac{2\xi}{\mathcal G}.
\end{equation*}
Thus $-\widehat{\boldsymbol v}$ is a strict descent direction. Setting $\xi=\rho\mathcal G$, $0<\rho<1$, in Eq.~\eqref{eq:complete-magnitude-l2-tail} gives, for either stated block,
\begin{equation*}
  S_{\mathrm{rel}}=\left\lceil \frac{4n\{1+\sqrt{2\log(1/\delta)}\}^2} {\rho^2\mathcal G^2}\right\rceil.
\end{equation*}
A fixed magnitude/phase allocation preserves the $n$ dependence for the full complex gradient. No uniform relative guarantee holds when the reference norm vanishes, $\mathcal G=0$.

For the theorem-adjacent separation, fix a parameter point and an active real coordinate $k$. Construct $O_k$ at that point and hold it fixed during differentiation:
\begin{equation*}
  O_k=I-(\ket{\psi^{\mathbb R}}-\ket{e_k^{\mathbb R}}) (\bra{\psi^{\mathbb R}}-\bra{e_k^{\mathbb R}}).
\end{equation*}
Orthonormality gives $O_k^\dagger=O_k$, $O_k^2=I$, and $O_k\ket{\psi^{\mathbb R}}=\ket{e_k^{\mathbb R}}$, so $\partial_{\theta_j}E_{O_k}=2\sqrt{g^{\mathbb R}_{k,k}}\delta_{jk}$ and $\chi_k=2$. If $0<g^{\mathbb R}_{k,k}<(\varepsilon/2)^2$, the zero vector has raw $\ell_\infty$ error below $\varepsilon$ but frame error two.

For a global magnitude record with $g^{\mathbb K}_{j,j}>0$, $Z_j\in\{\pm2\sqrt{g^{\mathbb K}_{j,j}}\}$, $\E Z_j=\partial_{\theta_j}E_O$, and $Z_j^2=4g^{\mathbb K}_{j,j}$. Unitarity gives $|\partial_{\theta_j}E_O|\leq2\sqrt{g^{\mathbb K}_{j,j}}$, while, for $\tau>0$,
\begin{equation*}
  \E\!\left[\left(\frac{Z_j}{g^{\mathbb K}_{j,j}}\right)^2\right] =\frac4{g^{\mathbb K}_{j,j}}, \qquad \E\!\left[\left(\frac{Z_j}{g^{\mathbb K}_{j,j}+\tau}\right)^2\right] =\frac{4g^{\mathbb K}_{j,j}}{(g^{\mathbb K}_{j,j}+\tau)^2} \leq\frac1\tau.
\end{equation*}
Small metric weights therefore suppress ordinary noise but condition an unregularized natural gradient. Damping or freezing them is an optimizer choice; the ordinary record vanishes when $g^{\mathbb K}_{j,j}=0$.

Since $n=\log_2(M+1)$ in the real chart and $n=\log_2(M+1)-1$ in the complex chart, fixed absolute coordinate accuracy and confidence give $S_{\mathrm{glob}}^{\mathbb K}=\bigO(\log\log M)$ and $S_{\mathrm{chk}}^{\mathbb K}=\bigO(\log M\log\log M)$; per-execution Hopf work and the materialized $M$-entry output remain separate.

\subsection{Assigned Hopf CNOT ledger and finite charges}
\label{app:proof-compiler-ledger}

Use the direct-angle count $C_{\mathrm H}$ of Sec.~\ref{subsec:assigned-ledger}; controlled-reflection, measurement, and classical costs remain separate. The controlled-rotation charge is
\begin{equation*}
  c_{R_y}(m) =
  \begin{cases}
    0, & m=0, \\
    2^{m+1}-2, & 1\leq m\leq4, \\
    16(m+1)-40, & m\geq5,
  \end{cases}
\end{equation*}
where the last line uses Ref.~\cite{vale2024decomposition} and open controls use exact $X$ conjugations. Set
\begin{align*}
  C_U&:=C_{\mathrm H}(\Uchk)=\sum_{d=0}^{n-1}2^dc_{R_y}(d),\\
  C_W&:=C_{\mathrm H}(\WR)=(N-1)c_{R_y}(n-1),\\
  C_{B_d}&:=C_{\mathrm H}(B_d^\dagger) =\sum_{h=d+1}^{n-1}2^hc_{R_y}(h).
\end{align*}
The final depth contributes $(N/2)c_{R_y}(n-1)$ and every suffix is contained in $\Uchk$, so
\begin{equation}
  C_W\leq2\left(1-\frac1N\right)C_U\leq2C_U, \qquad C_{B_d}\leq C_U,
  \label{eq:frame-depth-cost-ratio}
\end{equation}
with $C_U,C_W=\bigO(nN)$ for $n\geq2$. In the real chart set $(c_{\mathrm{fwd}},c_{\mathrm{fr}},c_d)=(C_U,C_W,C_{B_d})$; in the separated complex chart add the common phase charge $C_D:=C_{\mathrm H}(\Dph)$ to all three. Equation~\eqref{eq:frame-depth-cost-ratio} then gives
\begin{equation*}
  c_{\mathrm{fr}}^{\mathbb K}\leq2c_{\mathrm{fwd}}^{\mathbb K}, \qquad c_d^{\mathbb K}\leq c_{\mathrm{fwd}}^{\mathbb K},
\end{equation*}
which proves Eq.~\eqref{eq:hopf-reverse-forward-comparability}. Native-forward and integrated alternatives use the same functional but need not obey these finite constants; Appendix~\ref{app:four-qubit-ledger} gives their explicit four-qubit charges.

\subsection{Exact multiplexed recompilation}
\label{app:compiler-robustness}

The direct-angle ledger keeps coordinates and physical controls in one geometry. One exact multiplexed realization preserves the estimator while recombining those angles and therefore has its own ledger. Each forward depth $U_d=\sum_{r=0}^{2^d-1}\ketbra{r}{r}\otimes R_y(\theta_{2^d+r})\otimes I_{2^{n-d-1}}$ is uniformly controlled. Fix the unpruned Gray-code templates, with $2^d$ core CNOTs for $d\geq1$ and none at the root~\cite{bergholm2005uniformly}. Their forward charge is
\begin{equation*}
  C_{\mathrm{ucr}}(\Uchk)=\sum_{d=1}^{n-1}2^d=N-2.
\end{equation*}

For an addressed depth $d<n-1$, write a basis label as prefix $r$, target bit $q$, and lower suffix $z$. A reusable clean flag computes $F_d\ket{r,q,z}\ket0_f=\ket{r,q,z}\ket{[z=0]}_f$. Let $V_{d,f}$ apply the prefix-selected $R_y(\theta_{2^d+r})$ to the target when the flag is one and act as the identity when it is zero. This leaves the suffix unchanged, so $F_d^\dagger V_{d,f}F_d$ implements the addressed layer $V_d$ on a clean flag and returns the flag to zero. The final depth needs no flag. Composing the layers gives the full-frame contract
\begin{equation}
  \widetilde W_{\mathbb R}(\ket\varphi\ket0_f) =(\WR\ket\varphi)\ket0_f \qquad\text{for every }\ket\varphi.
  \label{eq:optimized-clean-flag-contract}
\end{equation}
The nonfinal multiplexors have $d+1$ controls, hence
\begin{equation*}
  C_{\mathrm{ucr-core}}(\WR) \leq\sum_{d=0}^{n-2}2^{d+1}+2^{n-1}=\frac32N-2.
\end{equation*}
Quadratic exact multi-controlled-$X$ synthesis contributes $2\sum_{m=1}^{n-1}\bigO(m^2)=\bigO(n^3)$ for computing and uncomputing the suffix predicates~\cite{barenco1995elementary}. Thus $C(\WR)=\bigO(N+n^3)=\bigO(N)$; exact diagonal synthesis also gives $C(\Dph)=\bigO(N)$~\cite{bullock2004smaller}.

For these fixed templates and $n\geq2$, a real scalar execution uses $N-2$ forward-core CNOTs and a global-record execution adds at most $\frac32N-2+\bigO(n^3)$ for the inverse frame. The matched per-execution ratio is bounded by
\begin{equation*}
  1+\frac{\frac32N-2+\bigO(n^3)}{N-2}=\bigO(1).
\end{equation*}
Adding the same nonnegative controlled-reflection charge preserves this bound. In the separated complex construction, the same $\bigO(N)$ phase-layer charge is added to both the forward and reverse blocks, again giving a constant-order ratio. The global method therefore retains its $\bigO(\log\log M)$ statistical overhead under this one exact recompilation. The theorem continues to use the direct-angle ledger; no compiler-independent finite ratio is asserted.

\subsection{Proof of the complete-gradient Hopf
  quantum-backpropagation theorem}
\label{app:proof-hopf-qbp}

\begin{proof}[Proof of Theorem~\ref{thm:hopf-quantum-backpropagation}]
  Equations~\eqref{eq:real-hopf-estimator-properties}, \eqref{eq:complex-magnitude-estimator-bound}, and \eqref{eq:complex-phase-estimator-properties} give unbiased global records for all real and complex coordinates, while Eq.~\eqref{eq:checkpoint-depth-gradient} gives the complete block at each selected depth. Their norm-two bounds and the allocations in Eq.~\eqref{eq:complete-chart-shot-allocations} yield the target $(\varepsilon,\delta)$ accuracy. Global outcomes are reused across all magnitude depths, whereas checkpoints use one stream per depth, giving Eq.~\eqref{eq:global-checkpoint-shot-scaling}.

  The matched scalar and gradient programs share the direct-angle forward compiler and exact block formulas in Eqs.~\eqref{eq:matched-scalar-hopf-work} and \eqref{eq:hopf-work-block-formulas}. Substituting Eq.~\eqref{eq:hopf-reverse-forward-comparability} gives Eqs.~\eqref{eq:hopf-global-work-ratio} and \eqref{eq:hopf-checkpoint-work-ratio}; with $c_{\mathrm{fwd}}^{\mathbb K}=\bigO(M\log M)$, this gives Eq.~\eqref{eq:hopf-absolute-work-scaling}.

  Shared live workspace, ancilla-free sequential reverse blocks, and $\bigO(M)$ classical output give constant storage overhead. Using $n=\log_2(M+1)$ in the real chart and $n=\log_2(M+1)-1$ in the complex chart gives the stated $\bigO(\log\log M)$ global and $\bigO(\log M\log\log M)$ checkpoint classifications.
\end{proof}

\end{document}